\documentclass[12pt]{article}

\usepackage{amssymb,amsfonts,amsthm,mathtools,bm,soul}
\usepackage{bbm}
\usepackage[svgnames]{xcolor}
\usepackage{pdfpages, sectsty}
\sectionfont{\color{Maroon}}
\subsectionfont{\color{Maroon}}
\subsubsectionfont{\color{Maroon}}
\usepackage{colortbl}
\usepackage{upgreek}
\usepackage{enumitem}
\usepackage{comment}
\usepackage{nicefrac}

\usepackage{tikz-cd}
\usepackage[round, longnamesfirst]{natbib}   
\usepackage{multirow}

\usepackage[margin=1.2in]{geometry} 
\usepackage{caption}
\usepackage{subcaption,setspace}

\usetikzlibrary{arrows}
\usepackage{hyperref}
\usepackage[noabbrev,capitalise,nameinlink]{cleveref}

\hypersetup{
pdftitle={Withholding Verifiable Information},
pdfauthor={Denis Shishkin, Maria Titova, and Kun Zhang},
pdfkeywords={Communication, Disclosure games, Commitment}
}
\hypersetup{
colorlinks,breaklinks,naturalnames,
citecolor=DarkBlue,urlcolor=Indigo,linkcolor=DarkBlue
}

\DeclareMathOperator\supp{supp}
\DeclareMathOperator\inter{int}
\DeclareMathOperator\conv{co}
\DeclareMathOperator\clos{cl}

\newcommand{\dd}{\mathrm{d}}
\newcommand{\ddd}{\, \mathrm{d}}
\newcommand{\beliefm}{\beta}

\newtheorem{theorem}{Theorem}
\newtheorem{proposition}{Proposition}
\newtheorem{lemma}{Lemma}
\newtheorem{claim}{Claim}

\newtheorem{observation}{Observation}
\newtheorem{corollary}{Corollary}

\theoremstyle{definition}

\newtheorem{definition}{Definition}

\newcommand{\sep}[0]{ \; |\; }

\newcommand{\cups}[0]{ \cup \ldots \cup }

\definecolor{MyBlue}{RGB}{0,91,148}
\definecolor{MyRed}{RGB}{200,15,62}
\definecolor{MyGreen}{RGB}{10,120,90}
\definecolor{MyOrange}{RGB}{255,165,79}

\usepackage[
    justification = centering,
    format        = plain,
    labelfont     = {bf,it},
    textfont      = it,
    belowskip     = -5pt,
    labelsep      = period
]{caption}

\setlist[enumerate, 1]{
    itemsep   = 0.05em,
    parsep    = 0.1em,
    partopsep = 0.1em,
    topsep    = 0.1em,
}

\setlist[itemize, 1]{
    itemsep   = 0.05em,
    parsep    = 0.1em,
    partopsep = 0.1em,
    topsep    = 0.1em,
    label     = $\bullet$,
}

\setlist[itemize, 2]{
    itemsep   = 0.15em,
    parsep    = 0.4em,
    partopsep = 0.45em,
    label     = --,
}

\setlist[itemize, 3]{
    itemsep   = 0.1em,
    parsep    = 0.4em,
    partopsep = 0.45em,
    label     = $\triangleright$
}

\newcommand{\drawB}[5]{
    \def\Bxstart{#1}
    \def\Bxend{#2}
    \def\Bylevel{#3}
    \def\Bheight{0.06}
    \def\Bycoord{\Bylevel*\Bheight+0.02}
    \def\Bcolor{#4}
    \def\Blabel{#5}
    \node[left,\Bcolor] at (0,{\Bycoord+\Bheight/2}) {{\footnotesize\Blabel}};
    \draw[fill=\Bcolor!15,draw=\Bcolor!80,thick] ({\Bxstart},{\Bycoord}) rectangle ({\Bxend},{\Bycoord+\Bheight});
}
\newcommand{\drawemptyB}[3]{
    \def\Bylevel{#1}
    \def\Bycoord{\Bylevel*\Bheight+0.02}
    \def\Bcolor{#2}
    \def\Blabel{#3}
    \node[left,\Bcolor] at (0,{\Bycoord+\Bheight/2}) {{\footnotesize\Blabel}};
}

\newcommand{\drawmean}[6]{
    \def\Bcolor{#6}
    \def\Bylevel{#5}
    \def\Bmean{{((#2-#1)*(#1+#2)+(#4-#3)*(#3+#4))/(2*(#2-#1+#4-#3))}}
    \def\Bycoord{\Bylevel*\Bheight+0.02}
    \draw[#6,dashed,line width = 0.75mm] ({\Bmean},{\Bycoord}) -- ({\Bmean},{\Bycoord+\Bheight});
}

\title{Withholding Verifiable Information\thanks{We thank
Nageeb Ali,
Hector Chade,
Vincent Crawford,
Yuhta Ishii,
Sergei Izmalkov,
Andreas Kleiner,
Qianjun Lyu,
Jo\~{a}o Ramos,
and various seminar and conference participants for their helpful comments and suggestions.}}
\author{Denis Shishkin\thanks{Department of Economics, University of California San Diego, \href{mailto:dshishkin@ucsd.edu}{\texttt{dshishkin@ucsd.edu}}.} \and Maria Titova\thanks{Department of Economics, Vanderbilt University, \href{mailto:motitova@gmail.com}{\texttt{motitova@gmail.com}}.} \and Kun Zhang\thanks{School of Economics, University of Queensland, \href{mailto:kun@kunzhang.org}{\texttt{kun@kunzhang.org}}.}}
\date{\today}

\begin{document}

\begin{titlepage}
\maketitle
\onehalfspacing  
\thispagestyle{empty}

\begin{abstract}
    We study a class of finite-action disclosure games in which the sender's preferences are state-independent and the receiver's optimal action depends only on the expected state. While receiver-preferred equilibria in these games involve full revelation, other equilibria are less well understood.
    We show that any equilibrium payoff can be obtained with a disclosure strategy corresponding to a partition with a laminar structure that allows pooling of nonadjacent states.
    In a sender-preferred equilibrium, such a structure balances inducing more sender-favorable actions with deterring deviations.
    Leveraging this insight, we identify conditions under which the sender does not benefit from commitment power. We then apply these results to study selling with quality disclosure and influencing voters.
\end{abstract}

\end{titlepage}

\section{Introduction}
A canonical prediction in disclosure games is \emph{unraveling} \citep[e.g.,][]{grossman1981informational,milgrom1981good}.
Specifically, if a sender can credibly prove that a state is highly favorable, he reveals it to induce a higher action from a receiver.
Once the most favorable states have been disclosed, slightly less favorable states must also be revealed to avoid being mistaken for worse ones, and this reasoning continues recursively until all states are revealed.
This full-revelation result hinges on a crucial assumption: the receiver's action space is sufficiently flexible, meaning that she can adjust her action continuously so that any marginal improvement in her belief leads to a strictly higher action.

However, this flexibility assumption often fails in settings where the receiver chooses among finitely many actions---for example, when a consumer decides between a few products or product versions 
or a policymaker chooses among a few policy alternatives.
In such cases, the receiver cannot finely adjust her action in response to small changes in her belief. 
Previous studies \citep[e.g.,][]{giovannoni2007secrecy,titova2025persuasion} have shown that this discreteness may prevent full unraveling and allow the sender to withhold some information---that is, create a scope for pooling. Yet little is known about precisely which states are pooled in equilibrium or about the limits of what can be achieved through verifiable disclosure.

In this paper, we study a disclosure game in which the receiver's preferred action is increasing in the expected state. The only essential difference from \cite{milgrom1981good} is that our receiver's action space is finite. We characterize the equilibrium payoff set, study the sender-preferred equilibrium payoff, and identify sufficient conditions under which the commitment payoff is achieved with verifiable disclosure. We illustrate our results with a motivating example.

\newcommand{\midstep}{p}

\paragraph{Example.} Consider a seller (he) promoting a product to a buyer (she) who chooses whether to buy nothing (action $1$), buy the product (action $2$), or buy the product bundled with an add-on (action $3$).
The players have a common prior that the product quality---which can be interpreted as either the sender's type or the state of the world---is uniformly distributed on $[0,1]$.
The buyer's payoff depends on the posterior expectation of product quality $\omega$ and is such that she optimally buys the bundle if $\mathbb{E}[\omega] \in [\frac34,1]$, only the standalone product if $\mathbb{E}[\omega] \in [\frac12,\frac34]$, and nothing otherwise.
The seller's profit is $0$ if he sells nothing, $1$ if he sells the product with the add-on, and $\midstep\in[0,1]$ if he sells the product only.
To persuade the buyer, the seller can disclose a piece of hard evidence about product quality after privately observing it.
In particular, he can send a message corresponding to any nonempty closed subset of $[0,1]$ containing the true quality $\omega$.
We focus on partitional disclosure strategies associated with some (ordered) partition $\{B_1,B_2,B_3\}$ of $[0,1]$ such that when the product quality is $\omega \in B_i$, the seller sends a message $B_i$, which is interpreted as a recommendation for the buyer to take action $i$.
To determine whether a partition can arise in an equilibrium, one needs to ensure that no player can profitably deviate.
First, the partition must satisfy \emph{obedience} in the sense that the buyer is willing to follow the recommendation.
Second, it must satisfy \emph{revelation proofness} in the sense that the seller is never willing to deviate by revealing the true quality.
One can show that these properties are not only necessary but also sufficient for a partitional strategy to be an equilibrium strategy.

First, note that the seller's equilibrium payoff is bounded from above by his \emph{commitment payoff} in this environment---that is, his maximal expected payoff in the case in which he can commit to disclosing information about $\omega$ using any experiment.
The commitment problem can be seen as a relaxation of the problem of maximizing the seller's payoff across equilibria because the former does not require the disclosure strategy to be revelation-proof.
Therefore, if the commitment payoff is attainable in some equilibrium, this equilibrium must be a seller-preferred equilibrium of the game.

Suppose first that $\midstep = 0$.
Given the observation above, we start by identifying the commitment solution.
Because $\midstep = 0$, the action space is effectively binary and the seller is maximizing the probability of selling the bundle.
Every commitment-optimal experiment corresponds to a partition given by $B_1 = [0,\frac12], B_2 = \varnothing, B_3=[\frac12,1]$ so that the buyer is indifferent between buying and not buying the bundle following message $B_3$.
It is easy to check that this partition is revelation-proof, as no seller type in $B_1$ can induce any action higher than $1$ by deviating.
We have thus found a seller-preferred equilibrium that attains the commitment payoff.

Suppose next that $\midstep = 0.5$.
In this case, there are three nontrivial actions available to the buyer, and the seller faces a trade-off between the likelihood of selling the standalone product and the likelihood of selling the bundle.
It turns out that the unique commitment-optimal partition is given by $B_1 = [0,\frac14], B_2 = [\frac{5}{16},\frac{11}{16}], B_3 = [\frac14,\frac{5}{16}]\cup[\frac{11}{16},1]$ (see \cref{fig:intro}).
Note two important properties of this partition.
First, while this partition must satisfy obedience, it does so only \emph{barely}, in the sense that the means of $B_2$ and $B_3$ are the lowest posterior means compatible with obedience.
Namely, their means are exactly the cutoffs $\frac12$ and $\frac34$ for actions $2$ and $3$, respectively.
Intuitively, if the obedience constraint were slack, the seller would be able to perturb the partition and shift the probability toward higher actions.
Second, in contrast to the case of $\midstep = 0$, the optimal partition here is not monotone in the sense that some elements are not intervals.
Instead, it is associated with a class of \emph{bi-pooling} distributions of posterior means, which are known to be optimal in a general class of linear persuasion problems \citep{kms,absy}.
Such bi-pooling partitions generalize monotone partitions by allowing pooling of nonadjacent types as follows.
Each partitional element $B_i$ is  an interval or consists of two intervals and ``nests'' a unique other interval $B_j$ in the sense that $\conv(B_i) \supseteq B_j$.
Intuitively, such bi-pooling partitions arise as commitment optima because among barely obedient partitions they always allow for optimally resolving the trade-off between the likelihoods of difference actions.
To determine whether the commitment payoff is attainable in equilibrium, we need to verify that this partition is revelation-proof.
It is sufficient to check that no type in $B_2$ can deviate by revealing their type and convincing the buyer to buy the bundle.
Indeed, since the highest type $\frac{11}{16}$ in $B_2$ is below the threshold $\frac34$ of selling a bundle, this partition is a seller-preferred equilibrium partition.

\begin{figure}[ht!]\footnotesize
    \begin{tikzpicture}[scale = 8]
        \draw[|-,thick] (0,0) node[below] {$0$} -- (0.5,0) node[midway, below] {buy nothing} -- (0.5,0);
        \draw[|-,thick] (0.5,0) node[below] {$\frac12$} -- (0.75,0) node[below] {$\frac34$} node[midway,below,align=center] {product \\ only};
        \draw[|-|,thick] (0.75,0) -- (1,0) node[below] {$1$} node[midway,below,align=center] {with\\ add-on};
        
        \drawB{0.5}{1}{2}{MyGreen}{{\(B_3\)}}
        \drawmean{0.5}{1}{0.5}{1}{2}{MyGreen}
    
        \drawemptyB{1}{MyRed}{{\(B_2\)}}
        
        \drawB{0}{0.5}{0}{MyBlue}{{\(B_1\)}}
        \drawmean{0}{0.5}{0}{0.5}{0}{MyBlue}

        \node[right] at (1.1,0.12) {$p=0$, with and without commitment};
    
    \end{tikzpicture}

    \begin{tikzpicture}[scale = 8]
        \draw[|-,thick] (0,0) node[below] {$0$} -- (0.5,0) node[midway, below] {buy nothing} -- (0.5,0);
        \draw[|-,thick] (0.5,0) node[below] {$\frac12$} -- (0.75,0) node[below] {$\frac34$} node[midway,below,align=center] {product \\ only};
        \draw[|-|,thick] (0.75,0) -- (1,0) node[below] {$1$} node[midway,below,align=center] {with\\ add-on};

        \drawB{1/4}{5/16}{2}{MyGreen}{{\(B_3\)}}
        \drawB{11/16}{1}{2}{MyGreen}{{\(B_3\)}}
        \drawmean{1/4}{5/16}{11/16}{1}{2}{MyGreen}

        \drawB{5/16}{11/16}{1}{MyRed}{{\(B_2\)}}
        \drawmean{5/16}{11/16}{5/16}{11/16}{1}{MyRed}
        
        \drawB{0}{1/4}{0}{MyBlue}{{\(B_1\)}}
        \drawmean{0}{1/4}{0}{1/4}{0}{MyBlue}

        \node[right] at (1.1,0.12) {$p=0.5$, with and without commitment};
    
    \end{tikzpicture}

    \def\bb{0.1909830056} 
    \begin{tikzpicture}[scale = 8]
        \draw[|-,thick] (0,0) node[below] {$0$} -- (0.5,0) node[midway, below] {buy nothing} -- (0.5,0);
        \draw[|-,thick] (0.5,0) node[below] {$\frac12$} -- (0.75,0) node[below] {$\frac34$} node[midway,below,align=center] {product \\ only};
        \draw[|-|,thick] (0.75,0) -- (1,0) node[below] {$1$} node[midway,below,align=center] {with\\ add-on};

        \drawB{\bb}{1/4}{2}{MyGreen}{{\(B_3\)}}
        \drawB{3/4}{1}{2}{MyGreen}{{\(B_3\)}}
        \drawmean{\bb}{1/4}{3/4}{1}{2}{MyGreen}
        
    
        \drawB{1/4}{3/4}{1}{MyRed}{{\(B_2\)}}
        \drawmean{1/4}{3/4}{1/4}{3/4}{1}{MyRed}
        
        \drawB{0}{\bb}{0}{MyBlue}{{\(B_1\)}}
        \drawmean{0}{\bb}{0}{\bb}{0}{MyBlue}

        \node[right] at (1.1,0.12) {$p=0.6$, without commitment};
    
    \end{tikzpicture}

    \def\bb{0.1909830056} 

\begin{tikzpicture}[scale = 8]
        \draw[|-,thick] (0,0) node[below] {$0$} -- (0.5,0) node[midway, below] {buy nothing} -- (0.5,0);
        \draw[|-,thick] (0.5,0) node[below] {$\frac12$} -- (0.75,0) node[below] {$\frac34$} node[midway,below,align=center] {product \\ only};
        \draw[|-|,thick] (0.75,0) -- (1,0) node[below] {$1$} node[midway,below,align=center] {with\\ add-on};

        \drawB{1/8}{11/64}{2}{MyGreen}{{\(B_3\)}}
        \drawB{53/64}{1}{2}{MyGreen}{{\(B_3\)}}
        \drawmean{1/8}{11/64}{53/64}{1}{2}{MyGreen}
        
    
        \drawB{11/64}{53/64}{1}{MyRed}{{\(B_2\)}}
        \drawmean{11/64}{53/64}{11/64}{53/64}{1}{MyRed}
        
        \drawB{0}{1/8}{0}{MyBlue}{{\(B_1\)}}
        \drawmean{0}{1/8}{0}{1/8}{0}{MyBlue}

        \node[right] at (1.1,0.12) {$p=0.6$, with commitment};
    
    \end{tikzpicture}

    \caption{Commitment-optimal and sender-preferred equilibrium partitions in the example. The dashed lines are the (conditional) means of the partitional elements.}
    \label{fig:intro}
\end{figure}

Suppose next that $\midstep = 0.6$.
The unique commitment-optimal partition can be shown to be given by a bi-pooling partition $B_1 = [0,\frac18], B_2 = [\frac{11}{64},\frac{53}{64}], B_3 = [\frac{1}{8},\frac{11}{64}]\cup[\frac{53}{64},1]$.
Compared to the previous case, here the profit from selling the standalone product expands and the set $B_2$ of types selling the standalone product shrinks.
In particular, there are now some types in $B_2$ that are above the bundle threshold of $\frac34$ and for whom revealing their type would be a profitable deviation.
Therefore, the above partition is not revelation-proof and the maximal equilibrium seller profit in the disclosure game is strictly below the commitment benchmark.
In this case, the seller-optimal equilibrium turns out to be associated with another bi-pooling partition given by $B_1 = [0,\frac{3-\sqrt{5}}{4}], B_2 = [\frac14,\frac34], B_3 = [\frac{3-\sqrt{5}}{4},\frac14]\cup [\frac34,1]$.
To see why the seller cannot do better in equilibrium, note that this partition has two constraints.
First, it is barely obedient for the same reason as the commitment optima described above.
Second, in contrast to the commitment solution, the upper bound of $B_2$ coincides with the threshold $\frac34$, indicating that the seller revelation-proofness constraint is binding.

Our model generalizes this example to any absolutely continuous prior distribution of the state, an ordered finite action space, and monotone preferences of the players.\footnote{The players are said to have monotone preferences if the sender's payoff is strictly increasing in the receiver's action and the receiver's preferred action is increasing in the expected state.}
Our first main result, \cref{thm:eq-payoff-set}, characterizes the sender's equilibrium payoff set by showing that every equilibrium payoff can be obtained in a {\it laminar partitional equilibrium}, that is, one associated with a laminar partition. The laminar property of a partition, introduced in \cite{candogan2023optimal}, generalizes the aforementioned bi-pooling property as follows.
In a bi-pooling partition, each element's convex hull may nest at most one other lower-indexed element, while in a laminar partition, it may nest any number of lower-indexed elements. 
The key step in the characterization is showing that laminar partitions are the most revelation-proof among all obedient partitions in the sense that any equilibrium partition can be transformed into a laminar equilibrium partition with the same sender payoff.

Our next result, \cref{thm:sender-pref-eq-characterization}, characterizes sender-preferred laminar equilibrium partitions and, more generally, provides a sharp equilibrium test for barely obedient laminar partitions. First, sender-preferred laminar equilibrium partitions are barely obedient in the sense that they leave no slack in the receiver's obedience constraints. \cref{thm:sender-pref-eq-characterization} also bounds their complexity by showing that action $i$ is recommended on at most $\max\{i-1,1\}$ intervals of states. Finally, for barely obedient laminar partitions, revelation proofness is simplified into a local condition: it suffices to check only a small subset of the non-lowest on-path actions and only at the highest state recommend each such action. In all other cases, revelation proofness is automatic.

Our remaining set of results discusses the attainability of the commitment payoff in the disclosure game. 
It is well known that generically every commitment problem admits a unique bi-pooling solution.
We establish in \cref{thm:id-sol} that, among all partitions yielding this solution, there is a unique barely obedient bi-pooling one.
Moreover, this partition is most resistant to violations of revelation proofness; consequently, the revelation proofness of this partition determines whether the commitment payoff can be achieved in a disclosure game (\cref{prop:bipooling-implementable-iff}). 
With binary actions, \cite{titova2025persuasion} show that every bi-pooling commitment solution is implementable. Our \cref{prop:sufficient} shows that for three or more actions, every bi-pooling commitment solution is implementable if the sender's utility is sufficiently convex in the cutoffs for the receiver's actions. In the above example with three actions, the sender's utility is sufficiently convex when $p$ is sufficiently low.

Our results suggest that full revelation may not be the only relevant outcome, even in settings where the sender can prove any true fact.
In particular, when the receiver has only finitely many actions, the emergence of a rich set of simple equilibria can help explain the variety of disclosure policies observed in practice.
Even within the class of laminar partitional equilibria, the same environment can sustain a fully revealing equilibrium, a sender-preferred equilibrium with substantially more withholding of information, and every payoff in between.
Our analysis of sender-preferred equilibria also informs the case of mandatory disclosure: in some environments, voluntary disclosure may lead to pooling of a large fraction of low states with high states and therefore a much lower receiver's payoff.

We apply these insights to study selling with quality disclosure and influencing voters. In the former setting, \cite{milgrom1981good} shows that when the buyer can purchase any fraction of the product, unraveling takes place and every equilibrium features full revelation. We show that if the buyer is restricted to purchasing integer units, the seller may be able to achieve the commitment payoff by withholding information. In the second application, we consider an expert who discloses verifiable information to a voter who chooses from three alternatives: the amended bill, the unamended bill, and the status quo. We demonstrate that the expert can be hurt even if, all else equal, the voter becomes more inclined toward the expert's most preferred alternative.

\paragraph{Related literature.}
This paper belongs to a growing body of literature that characterizes the equilibrium payoff set in disclosure games.\footnote{Disclosure games were introduced in \cite{gh80}, \cite{grossman1981informational}, and \cite{milgrom1981good}. For surveys of this literature, see \cite{milgrom2008persuasion}, \cite{dj10}, and \cite{benporath2025evidence}.} The most closely related paper is \cite{titova2025persuasion}, which studies a more general disclosure game with a finite number of receiver actions. We specialize their model by assuming the state space is the unit interval and the receiver has monotonic preferences that depend only on the expected state. These assumptions allow us to characterize sender-preferred equilibria, identify the limits of verifiable communication, and establish sufficient conditions on model primitives under which the sender achieves the commitment payoff. While we focus on environments in which the receiver's action set is more limited than it is in \citet{grossman1981informational} and \citet{milgrom1981good}---our receiver chooses from a finite set of actions---\citet{ali2024from} study settings with greater flexibility, where full revelation prompts an action that makes the sender no better off than inducing any other beliefs. Consequently, revelation proofness is not a concern. They provide conditions under which the set of equilibrium payoff profiles is virtually the same as the set of achievable payoff profiles under commitment. \cite{GieczewskiTitova} consider disclosure games with a general message mapping, propose an equilibrium selection criterion related to neologism proofness, and characterize the sender's ex-ante payoffs under this criterion when the sender has access to sufficiently rich stochastic evidence.\footnote{``Stochastic evidence'' means that the sender's available messages depend on both the state and chance. In our game, however, the set of available messages is determined solely by the state.}

Beyond expanding the work on disclosure games, our work contributes to the growing body of literature on the possibility of attaining the commitment payoff without full commitment in other communication environments, as in the cases of cheap talk \citep{lipnowski2020cheap,lipnowski2020equivalence}, repeated cheap talk \citep{best2024persuasion,kuvalekar2022goodwill,mathevet2022reputation,pei2023repeated}, informed information design \citep{PerezRichet2014, koessler2023informed,Zapechelnyuk2023},\footnote{One can also interpret these models as disclosure games in which the sender has access to stochastic evidence.} the ability to covertly revise a message generated by an experiment \citep{min2021bayesian,lrs}, the ability to covertly revise an experiment without affecting the marginal distribution over messages \citep{lin2024credible}, costly misreporting \citep{guo2021costly,nguyen2021bayesian}, disclosure following private experimentation \citep{arieli2025bayesian,dai2026bayesian}, and Bayesian persuasion under sender-worst equilibrium selection \citep{lipnowski2025perfect}.
In contrast to these studies, we study a one-shot communication game with verifiable information without any commitment.

In the literature, pooling of nonadjacent states has been obtained in settings with either sender commitment \citep[e.g.,][]{kms,absy}, the receiver's private information \citep[e.g.,][]{feltovich2002toocool,harbaugh2020false}\footnote{\cite{bederson2018incomplete} also provide empirical evidence.}, or both \citep[e.g.,][]{guo2019interval,candogan2023optimal}.
While our setting has neither sender commitment nor receiver private information, equilibria in our game feature a similar structure.
The most closely related paper to ours is \cite{candogan2023optimal}, which studies linear persuasion with one or more privately informed receivers. In that paper, laminar partitions arise as solutions to standard linear maximization problems with a mean-preserving contraction constraint, in which the receivers' incentive constraints are written as additional moment conditions.
Despite the fact that our revelation-proofness constraint cannot be written as a moment condition, the laminar property likewise plays a key role in our disclosure game because, as we show, laminar partitions turn out to be the most revelation-proof.

\section{The Model}

We consider the following {\it disclosure game} between the sender (he) and the receiver (she).\footnote{We model verifiable disclosure in the same way as it is modeled in the seminal papers of \cite{gh80}, \citet{grossman1981informational}, and \citet{milgrom1981good}.} The state space is $\Omega = [0,1]$, and the common prior belief $\mu_F \in \Delta\Omega$ is induced by a CDF $F$ that admits a strictly positive density $f$.
First, the sender learns the state. Then the sender communicates with the receiver using verifiable messages. Specifically, the sender's message space in state $\omega \in \Omega$ is $M(\omega) := \{m \in \mathcal{C} : \omega \in m\}$, where $\mathcal{C}$ is the collection of all nonempty closed subsets of $[0,1]$. Finally, the receiver observes the message, forms a posterior belief, and takes an action. 

The receiver's action space is $N = \{ 1,\ldots, n \}$, where $n > 1$. The receiver's optimal action depends only on her posterior mean, denoted by $x \in [0,1]$.
We assume that the receiver's preferences are monotone in the sense that action $i$ is optimal if and only if $x \in [\gamma_{i},\gamma_{i+1}] =: A_i$ for some cutoffs $0 = \gamma_{1} < \gamma_{2} < \cdots < \gamma_{n+1} = 1$.\footnote{One interpretation is that the receiver wants to match the state but is constrained by the number of available actions.} 
We also assume that the sender's state-independent payoff $u_i$ from the receiver taking action $i$ is strictly increasing in $i$.\footnote{We use ``increasing,'' ``smaller,'' and ``greater'' in the weak sense; ``strictly'' will be added whenever needed.}
Without loss, we normalize $u_1 = 0$.
Finally, let 
\[
v(x) =
\begin{cases}
u_i & \text{if } x \in [\gamma_{i},\, \gamma_{i+1})\ \text{for some}\ i \in N \smallsetminus \{ n \},\\
u_n & \text{if } x \in [\gamma_{n},\, \gamma_{n+1}]
\end{cases}
\]
denote the sender's \ul{value function}, which maps the receiver's posterior mean to the highest attainable sender payoff. By construction, $v(x)$ is upper semicontinuous.

We focus on perfect Bayesian equilibria of the disclosure game. An \ul{assessment} is a triple $(\sigma,\tau,\beliefm)$, where $\sigma\colon \Omega \to \Delta \mathcal{C}$ is the sender's strategy, $\tau\colon \mathcal{C} \to \Delta N$ is the receiver's strategy, and $\beliefm\colon \mathcal{C} \to \Delta \Omega$ is the receiver's belief system.\footnote{For a compact metric space $Y$, let $\Delta Y$ denote the set of all probability measures on the Borel subsets of $Y$. Endowed with the Hausdorff distance, $\mathcal{C}$ is a compact metric space.} An assessment $(\sigma,\tau,\beliefm)$ is a (perfect Bayesian) \ul{equilibrium} if
    \begin{enumerate}
        \item for every $\omega \in \Omega$, $\sigma(\omega)$ is supported on $\arg\max_{m \in M(\omega)} \sum_{i\in N} \tau(i \sep m) u_{i}$; \label{eq-cond-1}
        \item for every $m \in \mathcal{C}$, $\tau(i \sep m) > 0$ implies $\int_\Omega \omega\ddd \beliefm(\omega \sep m) \in A_i$; \label{eq-cond-2}
        \item 
        $\beliefm$ is obtained from $F$ given $\sigma$ using Bayes' rule; \label{eq-cond-3} 
        \item for every $m \in \mathcal{C}$, $\supp(\beliefm(m)) \subseteq m$. \label{eq-cond-4}
    \end{enumerate}
    
In words, the first condition requires that the sender choose verifiable messages that maximize his expected payoff. The second condition requires that the receiver chooses an action that is optimal given her posterior mean. The third condition requires that the receiver uses Bayes' rule to calculate posterior beliefs from the prior and the sender's strategy. The final condition requires that the receiver's belief system is consistent with disclosure: she deems impossible any state in which the observed (on- or off-path) message is unavailable.

\section{Equilibrium Analysis}\label{section:eq-analysis}

We begin the analysis by introducing the notion of a partition of the state space and its key properties. A sequence $\mathcal{B} := \{B_i\}_{i \in N} \subseteq \mathcal{C}$ of closed subsets of $[0,1]$ is an (ordered) \ul{partition} if $\bigcup_{i \in N} B_i = [0,1]$ and $\mu_F\left(B_i \cap B_j\right) = 0$ for all $i, j \in N$.
Note that we index partitions using the action set for notational convenience; we also allow partitional elements to have nonempty but null intersections.

We say that an assessment
$(\sigma,\tau,\beliefm)$ and a partition $\mathcal{B}$ are \ul{associated} if for each $i \in N$,
\begin{itemize}
    \item[(a)] $\sigma(B_i \sep \omega) = \mathbbm{1}(\omega \in B_i \text{ and } \omega \notin B_{i+1} \cups B_n)$;
    \item[(b)] $\tau(i \sep B_i) = 1$;
    \item[(c)] $\beliefm(\cdot \sep B_i) = \mu_F(\cdot \sep B_i)$.
\end{itemize}
In an assessment associated with a partition $\mathcal{B}$, the sender's strategy is to reveal which element of the partition the state belongs to by sending message $B_i$ when $\omega \in B_i$.\footnote{Whenever $\omega$ belongs to multiple partition elements, the sender sends the message corresponding to the one with the highest index.}
Such a message can be interpreted as a recommendation for the receiver to choose action $i$.
Then the receiver's strategy is to always follow the recommendation.
Finally, the receiver's posterior beliefs on the path are computed using Bayes' rule.
Note that a partition uniquely defines the players' on-path behavior: if a partition is associated with multiple assessments, all of those assessments differ only in the receiver's off-path beliefs and actions. If an assessment $(\sigma,\tau,\beliefm)$ associated with a partition $\mathcal{B}$ is an equilibrium, then we call $(\sigma,\tau,\beliefm)$ a \ul{partitional equilibrium} and $\mathcal{B}$ an \ul{equilibrium partition}. The following two properties are necessary and sufficient for $\mathcal{B}$ to be an equilibrium partition.\footnote{This result was first shown in \cite{titova2025persuasion} for a disclosure game with a slightly different message space. We formulate and prove this result for our disclosure game in \cref{lemma:eqm-partition} in the appendix.}
\begin{definition}\label{dfn:RP-obedience}
    A partition $\mathcal{B}$ is
    \begin{itemize}
        \item \ul{obedient} if $\mathbb{E}[\omega \sep \omega \in B_i] \in A_i$ for each $i \in N$. 
        \item \ul{revelation-proof} if $\omega \in B_i$ implies $\omega \in A_1 \cups A_i =  [0,\gamma_{i+1}]$ for each $i \in N$.
    \end{itemize}
\end{definition}

In words, obedience requires that the receiver indeed prefers to take action $i$ when she learns that $\omega \in B_i$, that is, after message $B_i$. Revelation proofness ensures that fully revealing the state is not a profitable deviation for the sender. Indeed, in the disclosure game, the sender has the option to fully reveal the state by sending message $\{ \omega \}$ with probability one, thus convincing the receiver to take the action that she would take knowing $\omega$. Thus, if $\omega \in B_i$---in other words, the partition prescribes that the receiver takes action $i$ when the realized state is $\omega$---then the receiver prefers to take {\it at most} action $i$ when fully informed. 

Next we define a key structural property of partitions. Let $\clos(\cdot)$ and $\conv(\cdot)$ denote the closure and the convex hull operators, respectively.

\begin{definition}
    A partition $ \mathcal{B} $ is \ul{laminar} if for each $i \in N$,
    \begin{equation} \label{eqn:laminar-definition}
    B_i = \mathrm{cl}\left(\conv{(B_i)} \, \Big\backslash \, \bigcup_{j < i} \conv{(B_j)}\right).
    \end{equation}    
    If an equilibrium is associated with a laminar partition $\mathcal{B}$, we call this equilibrium a \ul{laminar partitional equilibrium}.
\end{definition}

The intuitive interpretation of the laminar property is that each element $B_i$ either has no ``gaps'' and is therefore an interval or has ``gaps'' and these gaps belong to lower-indexed partitional elements $B_1,\ldots,B_{i-1}$ only.
This implies that for each pair $B_i$ and $B_j$ with $i>j$, the convex hull $\conv(B_i)$ either nests or has a null intersection with $\conv(B_j)$.\footnote{
    In combinatorics, each pair of elements in a laminar set family is either nested or disjoint.
    We borrow the term from \cite{candogan2023optimal}, in which the definition of a laminar partition is similar to ours but allows for any nesting partial order.
    Their definition is essentially equivalent to ours for finite partitions, in the sense that our partition comes with an exogenous linear order that must complete the nesting partial order.
}
When $\conv(B_i) \supseteq \conv(B_j)$, we say that action $i$ \emph{nests} action $j$.
We illustrate this intuition in \cref{fig:exlaminar}.

The following observation is a direct consequence of the definition of a laminar partition. 

\begin{observation} \label{obs:laminar-i-intervals}
    If $\mathcal{B}$ is a laminar partition, then each $B_i$ is either an empty set or a union of at most $i$ intervals.
\end{observation}

\begin{figure}[ht!]
\begin{subfigure}{0.45\textwidth}
    \centering
       \begin{tikzpicture}[scale = 6.5]

        \draw[|-|,thick] (0,0) node[below] {$0$} -- (1,0) node[below] {$1$};

        \drawB{0.2}{0.3}{3}{MyOrange}{{\(B_4\)}}
        \drawB{0.4}{0.6}{3}{MyOrange}{}
        \drawB{0.7}{1}{3}{MyOrange}{}

        \drawB{0.6}{0.7}{2}{MyGreen}{{\(B_3\)}}

        \drawB{0.3}{0.4}{1}{MyRed}{{\(B_2\)}}
        
        \drawB{0}{0.2}{0}{MyBlue}{{\(B_1\)}}
    
    \end{tikzpicture}
    \caption{Only $B_4$ has gaps, but they belong to the lower-indexed $B_2$ and $B_3$.}
    \label{fig:exislaminar}
\end{subfigure}
~
\begin{subfigure}{0.45\textwidth}
    \centering
        \begin{tikzpicture}[scale = 6.5]

        \draw[|-|,thick] (0,0) node[below] {$0$} -- (1,0) node[below] {$1$};

        \drawB{0.5}{0.6}{3}{MyOrange}{$B_4$}
        \drawB{0.7}{1}{3}{MyOrange}{}

        \drawB{0.4}{0.5}{2}{MyGreen}{$B_3$}
        \drawB{0.6}{0.7}{2}{MyGreen}{}

        \drawB{0.2}{0.4}{1}{MyRed}{$B_2$}
        
        \drawB{0}{0.2}{0}{MyBlue}{$B_1$}

    \end{tikzpicture}
    \caption{$B_3$ has a gap that belongs to the higher-indexed $B_4$.}
    \label{fig:exisnotlaminar}
\end{subfigure}
\caption{A laminar partition in panel (a) and a non-laminar partition in panel (b).}
\label{fig:exlaminar}
\end{figure}

One prominent example of an obedient and revelation-proof laminar partition is $\mathcal{A} := \{ A_i \}_{i \in N} = \{ [\gamma_{i},\gamma_{i+1}] \}_{i \in N}$, which we call the \ul{fully informative} partition. 
Note that while $\mathcal{A}$ is not fully revealing, it provides the minimal information necessary for the receiver to take her complete-information optimal action in each state.

\subsubsection*{Equilibrium Payoff Set}

Next we turn to the characterization of the equilibrium payoff set.
We say that an equilibrium is \ul{sender-preferred} if it yields the highest ex-ante payoff for the sender across all equilibria. Analogously, we refer to the ex-ante payoff-minimizing equilibrium as \ul{sender-worst}. 
Our first main result shows that every equilibrium ex-ante payoff of the sender is achievable in an equilibrium associated with an obedient and revelation-proof laminar partition.\footnote{In Supplementary Appendix \ref{refinement}, we show that all partitional equilibria in our model satisfy standard refinements, such as the never-a-weak-best-response (NWBR) criterion \citep{chokreps1987} and the Grossman–Perry–Farrell refinement \citep{bertomeu2018verifiable}.}

\begin{theorem}\label{thm:eq-payoff-set}
\begin{enumerate}[label=(\alph*)]
    \item There exists a sender-worst equilibrium that is associated with the fully informative partition $\mathcal{A}$ and yields  $\underline{V} \coloneqq \sum_{i \in N} \mu_F( A_i ) u_i$. \label{thm1-sw}
    \item There exists a sender-preferred equilibrium that is associated with an obedient and revelation-proof laminar partition and yields $\overline{V} > \underline{V}$. \label{thm1-sp}
    \item For any $V \in [\underline{V},\overline{V}]$, there exists an equilibrium associated with an obedient and revelation-proof laminar partition. \label{thm1-range}
\end{enumerate}
\end{theorem}

We provide the intuition for \cref{thm:eq-payoff-set} below. Part \ref{thm1-sw} is straightforward: as in most disclosure games, there exists an equilibrium in which the receiver acts as if she is fully informed. In our setting, that equilibrium is associated with the fully informative partition $\mathcal{A}$. The sender's ex-ante payoff cannot fall below $\underline{V}$, or else the sender will have a profitable deviation toward fully revealing some state.

Part (b) follows from two key observations.
First, note that it is without loss to focus on partitional equilibria.
This is because in a sender-preferred equilibrium, the receiver breaks ties in favor of the sender.
Hence, a single action is taken with probability one in each state.

Second, we show that for any equilibrium partition there exists an equilibrium {\it laminar} partition that induces the same posterior mean distribution (PMD) and hence yields the same ex-ante payoff to the sender. 
We first observe that any partition induces a PMD with support on at most $n$ points (posterior means); using the techniques from \cite{candogan2023optimal}, we show that any such PMD can be induced by a {\it laminar} partition. We further show that if a PMD is induced by an obedient and revelation-proof partition, then the laminar partition that induces the same PMD is guaranteed to be obedient and revelation-proof. 
In this sense, laminar partitions are the most revelation-proof partitions. 

\begin{figure}[ht!]
    \centering

        \begin{tikzpicture}[scale = 8]

  \draw[|-|,thick] (0,0) node[below] {$0$} -- (1,0) node[below] {$1$};

  \drawB{0}{0.2}{1}{MyRed}{{\(B_2\)}}
  \drawB{0.7}{1}{1}{MyRed}{{\(B_2\)}}
  \drawB{0.2}{0.7}{0}{MyBlue}{{\(B_1\)}}
  \draw (0.75,0.02) -- (0.75,-0.02) node[below] {$\gamma_2$};

\drawmean{0.2}{0.7}{0.2}{0.7}{0}{MyBlue}
\drawmean{0}{0.2}{0.7}{1}{1}{MyRed}











\end{tikzpicture}

        \begin{tikzpicture}[scale = 8]

  \draw[|-|,thick] (0,0) node[below] {$0$} -- (1,0) node[below] {$1$};

  \drawB{0.1}{0.35}{1}{MyRed}{{\(W_2\)}}
  \drawB{0.75}{1}{1}{MyRed}{{\(W_2\)}}
  \drawB{0}{0.1}{0}{MyBlue}{{\(W_1\)}}
  \drawB{0.35}{0.75}{0}{MyBlue}{{\(W_1\)}}
  \draw (0.75,0.02) -- (0.75,-0.02) node[below] {$\gamma_2$};

\drawmean{0}{0.1}{0.35}{0.75}{0}{MyBlue}
\drawmean{0.1}{0.35}{0.75}{1}{1}{MyRed}
  










\end{tikzpicture}

    \caption{Two partitions, $\mathcal{B}$ (laminar) and $\mathcal{W}$ (non-laminar), that induce the same PMD and whose corresponding elements have the same prior mass. If $\mathcal{W}$ is revelation-proof, then so is $\mathcal{B}$.}
    \label{fig:most-RP}
\end{figure}

We illustrate the intuition behind laminar partitions as the most revelation-proof ones in \cref{fig:most-RP}, which presents the simplest case, where $n=2$. Consider two partitions, $\mathcal{B}$ (laminar) and $\mathcal{W}$ (non-laminar), that induce the same PMD and whose corresponding elements have the same prior mass; that is, $\mathbbm{E}[\omega \sep \omega \in B_i] = \mathbbm{E}[\omega \sep \omega \in W_i]$ and $\mu_F(B_i) = \mu_F(W_i)$ for all $i \in N$. Since $\mathcal{B}$ is laminar, its lowest-indexed element $B_1$ is an interval. If $W_1$ is not an interval, the only way to match $B_i$'s prior mass and expectation is to have $\max B_1 \leq \max W_1$. Consequently, if $\mathcal{W}$ is revelation-proof (which, in the case of two actions, reduces to $\max W_1 \leq \gamma_2$), so is $\mathcal{B}$.
 
Combining the two aforementioned observations, to solve for the sender-preferred equilibrium payoff, one can maximize the sender's expected payoff across all obedient and revelation-proof laminar partitions as follows:
\begin{equation}\label{problem:sender-preferred-laminar}
    \begin{aligned}
    \overline{V}  = &\max_{ \mathcal{B} }   \sum_{i \in N} \mu_F(B_i) u_i \\
    &\text{subject to}\colon \mathcal{B} \text{ is an obedient and revelation-proof laminar partition.}
\end{aligned}
\end{equation}
We show that problem \eqref{problem:sender-preferred-laminar} admits a solution,\footnote{The existence of a solution follows from the extreme value theorem. 
To apply it, we endow the space of laminar partitions with a topology that makes the sender's payoff continuous and the constraint set compact. See \cref{lemma:payoff-max-laminar} and its proof in the appendix for details.} which implies that a sender-preferred equilibrium exists. To show that $\overline{V} > \underline{V}$, we construct an obedient and revelation-proof (and hence equilibrium) laminar partition that yields a strictly higher payoff than $\underline{V}$. For example, take the fully revealing partition $\{ [\gamma_i,\gamma_{i+1}] \}_{i\in N}$ and ``move'' the interval $[0,\varepsilon]$ from the first element of the partition to the $n$-th one. If $\varepsilon >0$ is sufficiently small, the resulting partition will remain obedient, revelation-proof, and laminar, but the sender's ex-ante utility in the associated equilibrium will be strictly higher than $\underline{V}$. Therefore, the sender's ex-ante payoff in his most preferred equilibrium will also exceed $\underline{V}$.

To prove Part \ref{thm1-range}, we show that for any payoff between $\underline{V}$ and $\overline{V}$ there exists a laminar equilibrium partition that blends the sender-worst (fully informative and laminar) partition $\mathcal{A}$ and the sender-preferred laminar equilibrium partition $\mathcal{B}$, and this blend yields the payoff.

\subsubsection*{Sender-Preferred Equilibrium}

Next we characterize the sender-preferred laminar partitional equilibrium. We say that action $i$ is \ul{skipped} in partition $\mathcal{B}$ if $B_i$ is a null set and \ul{unskipped} otherwise. We refer to partition $\mathcal{B}$ as \ul{barely obedient} if $\mathbbm{E}[\omega \sep \omega \in B_i] = \gamma_i$ for all but the lowest unskipped $i \in N$.

\begin{theorem} \label{thm:sender-pref-eq-characterization}
    Every sender-preferred laminar equilibrium partition $\mathcal{B}$ is barely obedient, and $B_i$ is the union of at most $\max\{i-1,1\}$ closed intervals for all $i \in N$. 
    Furthermore, a barely obedient laminar partition is an equilibrium partition if and only if $\max B_i \leq \gamma_{i+1}$ for any unskipped action $i < n$ such that action $i+1$ either nests $i$ or is skipped.
\end{theorem}

The first part of \cref{thm:sender-pref-eq-characterization} characterizes problem \eqref{problem:sender-preferred-laminar}'s solution---that is, an obedient and revelation-proof laminar partition that maximizes the sender's ex-ante payoff. First, $\mathcal{B}$ must be barely obedient; otherwise, one could pool additional low states with high states without violating obedience or revelation proofness, thereby obtaining an equilibrium with a strictly higher ex-ante sender payoff. Specifically, if $j$ is the lowest unskipped action, one can reassign a subset of $B_j$ of strictly positive measure to $B_k$ for some $k > j$. To show that each partitional element is the union of at most $\max\{i-1,1\}$ closed intervals, we first note that $B_j = [0,d]$ for some $d \le \gamma_{j+1}$ for the lowest unskipped action $j$. That is, the lowest unskipped action is an interval that is not in the convex hull of any other partitional elements. Then each partitional element with index $k \geq 2$ must be a union of at most $k-1$ intervals by the definition of a laminar partition.

Next, note that revelation proofness is equivalent to $\max B_i \le \gamma_{i+1}$ for each action $i$. 
The second part of \cref{thm:sender-pref-eq-characterization} shows that for a barely obedient laminar partition, this constraint is only binding for an action in two specific scenarios, illustrated in \cref{fig:RP-violations}: either the next-highest action is skipped or the action's corresponding partitional element is nested within the convex hull of the partitional element for the next-highest action. In all other instances, revelation proofness is automatically satisfied. This finding further emphasizes that the laminar structure makes revelation proofness easier to fulfill.

\begin{figure}[ht!]
    \centering
    \begin{subfigure}[b]{0.495\textwidth}
        \centering
        \begin{tikzpicture}[scale = 6.5]

  \draw[|-|,thick] (0,0) node[below] {$0$} -- (1,0) node[below] {$1$};

  \drawB{0.6}{1}{2}{MyGreen}{{\(B_3\)}}
  \drawemptyB{1}{MyRed}{{\(B_2\)}}
  \drawB{0}{0.6}{0}{MyBlue}{{\(B_1\)}}
  \draw (0.4,0.02) -- (0.4,-0.02) node[below] {$\gamma_2$};
  \draw (0.8,0.02) -- (0.8,-0.02) node[below] {$\gamma_3$};










\end{tikzpicture}
        \caption{If action $2$ is skipped and $\gamma_2 < \max B_1$, then full revelation induces the higher skipped action $2$ for all $\omega \in [\gamma_1,\max B_1) \subseteq B_1$.}
        \label{fig:nonRP-left}
    \end{subfigure}%
    ~
    \begin{subfigure}[b]{0.495\textwidth}
        \centering
        \begin{tikzpicture}[scale = 6.5]

  \draw[|-|,thick] (0,0) node[below] {$0$} -- (1,0) node[below] {$1$};

  \drawB{0}{0.2}{2}{MyGreen}{{\(B_3\)}}
  \drawB{0.7}{1}{2}{MyGreen}{{\(B_3\)}}
  \drawB{0.2}{0.7}{1}{MyRed}{{\(B_2\)}}
  \drawemptyB{0}{MyBlue}{{$B_1$}}
  \draw (0.5,0.02) -- (0.5,-0.02) node[below] {$\gamma_3$};

  








\end{tikzpicture}
        \caption{If $\gamma_3 < \max B_2$ and $\conv(B_2) \subseteq \conv(B_3)$, then full revelation induces the higher action $3$ for all $\omega \in [\gamma_2,\max B_2) \subseteq B_2$.}
        \label{fig:nonRP-right}
    \end{subfigure}
    \caption{Violations of revelation proofness ruled out in \cref{thm:sender-pref-eq-characterization}.}
    \label{fig:RP-violations}
\end{figure}

As in \cite{candogan2023optimal}, the laminar structure emerges in a sender-preferred equilibrium of our game due to incentive constraints. In \cite{candogan2023optimal}, where the receiver is privately informed, the laminar structure prevents the receiver from misreporting her private information. In our setting, however, it optimally balances the trade-off between deterring the sender's deviations and inducing the most desirable action distribution. 

As discussed after \cref{thm:eq-payoff-set}, laminar partitions are the most revelation-proof among all obedient partitions. However, interval partitions---a special case of laminar partitions---often fail to be sender-preferred equilibrium partitions. This underscores the importance of pooling nonadjacent states in inducing the most desirable distribution over actions subject to obedience and revelation proofness. To illustrate, recall our introductory example with $p=0.6$. The interval partition that generates the highest sender's ex-ante payoff is given by $B_1^I = [0,\frac14], B_2^I = [\frac14,\frac34], B_3^I = [\frac34,1]$. Since $\mathbbm{E}[\omega \sep \omega \in B_3^I] > \gamma_3 = \frac{3}{4}$, we can pool the states from the top of $B_1^I$ with $B_3^I$ until the partition becomes barely obedient. This process yields a sender-preferred equilibrium partition, $B_1 = [0,\frac{3-\sqrt{5}}{4}], B_2 = [\frac14,\frac34], B_3 = [\frac{3-\sqrt{5}}{4},\frac14]\cup [\frac34,1]$.

\section{When Is Commitment Payoff Achievable?} \label{s:imp}
While the extent to which the sender benefits from verifiable communication depends on the specific parameters, an upper bound on the sender's payoff is given by his commitment payoff---that is, his payoff when he can \emph{commit} to what messages to send in each state. In this section, we identify conditions under which the sender can attain his commitment payoff in an equilibrium of the disclosure game. 

\subsection{Commitment Benchmark} \label{ss:idb}
We start by introducing the commitment problem, or information-design problem, as a benchmark. In this problem, the sender can commit to any experiment that reveals information about the state. An \ul{experiment} is a mapping $\chi\colon \Omega \to \Delta(S)$, where $S$ is a sufficiently rich signal space.
For each state $\omega \in \Omega$, a signal $s \in S$ realizes according to $\chi(\omega)$. Because the receiver's optimal action only depends on the expected state, it is without loss to restrict attention to the class of experiments in which $S = [0,1]$ and each $s \in S$ is calibrated to equal the induced posterior mean: $s = \mathbb{E}[\omega \sep s]$.
Such a calibrated experiment $\chi$ induces the PMD with CDF $G(x) = \int_0^1 \int_0^x \dd \chi(s \sep \omega) \dd F(\omega)$.

It is well known that a PMD $G$ is induced by some experiment if and only if $G$ is a mean-preserving contraction of the prior CDF $F$.\footnote{See, for example, \citet{gk16} and \citet{kolotilin2018optimal}. A distribution $G \in \Delta([0,1])$ is a mean-preserving contraction of $F$ if $\int_{0}^{x} G(s) \mathrm{d} s \leq \int_{0}^{x} F(s) \mathrm{d} s$ for all $x \in [0,1]$, where the inequality binds at $x=1$.\label{fnc}} Consequently, the commitment problem can be stated as a maximization of the expected value $v$ with respect to the PMD,
\begin{equation} \label{id_problem}
	\max _{G \in \text{MPC}(F)} \int_{0}^{1} v(x) \ddd G(x),
\end{equation}
where $\text{MPC}(F)$ is the set of all mean-preserving contractions of $F$.
We call any solution to problem \eqref{id_problem} a \ul{commitment solution} and call the value of problem \eqref{id_problem} the \ul{commitment payoff}.
Clearly, the commitment payoff is an upper bound of the sender's equilibrium payoff in the disclosure game.

Finally, we say that a commitment solution $G$ is implementable with verifiable messages, or simply \ul{implementable} for short, if there is an equilibrium in which the sender's strategy induces the receiver's PMD $G$.

We say that an experiment $\chi$ is associated with a partition $\mathcal{B}$ if it maps almost every $\omega \in B_i$ into the degenerate distribution centered on $\mathbb{E}[\omega \sep \omega\in B_i]$.
In other words, such an experiment discloses only which element of the partition the state belongs to.
In this case, we also say that the induced PMD $G$ is associated with this partition.

Next we define a refinement of the laminar property, which is key for optimality under commitment.

\begin{definition}
    A laminar partition is a \ul{bi-pooling partition} if for every $i\in N$, either $B_i$ is an interval or there exists a unique $j < i$ and $B_j \subseteq \conv(B_i)$.
\end{definition}

\begin{figure}[!htb]
\begin{subfigure}{0.45\textwidth}
    \centering
        \begin{tikzpicture}[scale = 6.5]

        \draw[|-|,thick] (0,0) node[below] {$0$} -- (1,0) node[below] {$1$};

        \drawB{0.4}{0.5}{3}{MyOrange}{$B_4$}
        \drawB{0.7}{1}{3}{MyOrange}{}

        \drawB{0.5}{0.7}{2}{MyGreen}{$B_3$}

        \drawB{0.2}{0.4}{1}{MyRed}{$B_2$}
        
        \drawB{0}{0.2}{0}{MyBlue}{$B_1$}

    \end{tikzpicture}
    \caption{$B_1$, $B_2$, and $B_3$ are intervals, and $\conv{(B_4)}$ only contains $B_3$.}
    \label{fig:exisbp}
\end{subfigure}
~
\begin{subfigure}{0.45\textwidth}
    \centering
       \begin{tikzpicture}[scale = 6.5]

        \draw[|-|,thick] (0,0) node[below] {$0$} -- (1,0) node[below] {$1$};

        \drawB{0.2}{0.3}{3}{MyOrange}{{\(B_4\)}}
        \drawB{0.4}{0.6}{3}{MyOrange}{}
        \drawB{0.7}{1}{3}{MyOrange}{}

        \drawB{0.6}{0.7}{2}{MyGreen}{{\(B_3\)}}

        \drawB{0.3}{0.4}{1}{MyRed}{{\(B_2\)}}
        
        \drawB{0}{0.2}{0}{MyBlue}{{\(B_1\)}}
    
    \end{tikzpicture}
    \caption{$B_4$ is not an interval, and $B_2, B_3 \subseteq \conv{(B_4)}$.}
    \label{fig:exisnotbp}
\end{subfigure}
\caption{A bi-pooling partition in panel (a) and a laminar partition that is not a bi-pooling partition in panel (b).}
\label{fig:exbipool}
\end{figure}

While laminar partitions allow any $B_j$ with $j < i$ to be nested within $\conv(B_i)$, bi-pooling partitions allow only for a single such $B_j$ to be nested within $\conv(B_i)$ (see \cref{fig:exbipool}).

The following result regarding bi-pooling partitions is a direct consequence of the results in \cite{kms}, \cite{candogan2022persuasion}, and \cite{absy}. Say that the communication environment is \ul{generic} if no three elements of the collection of points $\{(\gamma_i, u_i)\}_{i \in N}$ are collinear.

\begin{theorem}\label{thm:id-sol}~
    There exists a commitment solution associated with a unique barely obedient bi-pooling partition.
    Moreover, the commitment solution is unique in generic communication environments.
\end{theorem}

\cref{thm:id-sol} indicates that, despite the simplicity of bi-pooling partitions, there always exists a commitment solution associated with a bi-pooling partition. Moreover, under mild conditions, the \emph{unique} commitment solution is associated with a bi-pooling partition.

\subsection{Characterizing Implementability} \label{subsection:imp}
The following result characterizes the implementability of a commitment solution associated with a bi-pooling partition.

\begin{proposition} \label{prop:bipooling-implementable-iff} 
    Let $G$ be a commitment solution associated with a bi-pooling partition $\mathcal{B}$.
    Then $G$ is implementable if and only if $\mathcal{B}$ is revelation-proof.
\end{proposition}

It follows directly from \autoref{thm:id-sol} and \autoref{prop:bipooling-implementable-iff} that in a generic communication environment, the unique commitment solution is implementable if and only if the associated bi-pooling partition is revelation-proof.

The ``if'' part of \cref{prop:bipooling-implementable-iff} is a direct consequence of Theorem 2 in \cite{titova2025persuasion}, which states that for a partition associated with a commitment solution, implementability is equivalent to revelation proofness. Our primary contribution in \cref{prop:bipooling-implementable-iff} is to show that of all sender strategies that induce a commitment solution, the one associated with a bi-pooling partition has the best shot at being an equilibrium strategy. Roughly, this stems from the fact that bi-pooling partitions are a special case of laminar partitions, which exhibit a similar property.

\autoref{prop:bipooling-implementable-iff} is useful in that it suggests a ``guess and verify'' approach to finding the sender-preferred equilibrium. First, one finds the commitment solution using standard information-design methods. Second, one identifies the associated bi-pooling partition. If this partition proves to be revelation-proof, then the sender-preferred equilibrium has been successfully identified.

The next result, which is a corollary of \cref{prop:bipooling-implementable-iff} and \cref{thm:sender-pref-eq-characterization}, goes one step further: it reveals the exact features of bi-pooling partitions that fail revelation proofness and hence prevent the commitment solutions from being implementable. 

\begin{corollary} \label{cor:bip_imp_2}
    Let $G$ be a commitment solution associated with a bi-pooling partition $\mathcal{B}$.
    Then $G$ is implementable if and only if $\mathcal{B}$ is such that $\max B_i \leq \gamma_{i+1}$ for any unskipped action $i < n$ such that action $i+1$ either nests $i$ or is skipped.
\end{corollary}

Compared to the definition of revelation proofness, the condition in \cref{cor:bip_imp_2} is easier to verify when determining whether a commitment solution is implementable.
It also allows us to identify sufficient conditions under which commitment has no value in the subsequent \cref{subsection:sufficient-cond}. 

To illustrate this result, recall our introductory example with $A = \{1,2,3\}$, $u_1 = 0$, $u_2 = p$, $u_3 = 1$, $\gamma_2 = \frac{1}{2}$, $\gamma_3 = \frac{3}{4}$, and uniform $F$.
When $p = 0$, there is a unique commitment solution associated with a bi-pooling partition given by $B_1 = [0,\frac12], B_2 = \varnothing, B_3 = [\frac12,1]$.
Note that the condition in \cref{cor:bip_imp_2} only applies to action $1$ and is satisfied because $\max B_1 = \frac12 \leq \frac12 = \gamma_2$.
In other words, the commitment-optimal partition is revelation-proof because the sender cannot induce an action higher than $1$ by fully revealing the state when it is in $B_1$.

When $p = 0.5$, there is a unique commitment solution associated with a bi-pooling partition given by $B_1 = [0,\frac14], B_2 = [\frac{5}{16},\frac{11}{16}], B_3 = [\frac14,\frac{5}{16}]\cup [\frac{11}{16},1]$.
Note that the condition in \cref{cor:bip_imp_2} only applies to action $2$ and is satisfied because $\max B_2 = \frac{11}{16} \leq \frac34 = \gamma_3$.

\cref{prop:bipooling-implementable-iff} and \cref{cor:bip_imp_2} might suggest that implementing a commitment solution associated with a bi-pooling partition is relatively simple: one needs only to ensure that no action is recommended more often than revelation proofness allows. Indeed, \cite{titova2025persuasion} showed that when a receiver is choosing between two actions, revelation proofness is automatically satisfied. When there are three or more actions, however, the restriction imposed by revelation proofness can be substantial.
We illustrated this in our introductory example with $p = 0.6$.
In this case, the commitment-optimal bi-pooling partition is given by $B_1 = [0,\frac18], B_2 = [\frac{11}{64},\frac{53}{64}], B_3 = [\frac18,\frac{11}{64}]\cup [\frac{53}{64},1]$.
Note that in this case the condition in \cref{cor:bip_imp_2} is violated for action $2$, which is nested by action $3$, but $\max B_2 = \frac{53}{64} > \frac34 = \gamma_3$.

When there are two actions, the sender's sole objective is to maximize the probability that the ``high action'' $2$ is played. 
This in turn suggests that revelation proofness is never an issue: in any state in which action $2$ is played under complete information, there is no reason to recommend action $1$. However, when there are three actions, as \cite{gk16} note, the sender in the commitment problem faces a trade-off between inducing actions $2$ and $3$. When $\midstep$ is high, the gap between $u_2$ and $u_3$ is significantly smaller than that between $u_1$ and $u_2$, and hence it is more profitable to induce action $2$ more often: to guarantee obedience, recommending action $3$ more often must come with action $1$ being played more frequently. Consequently, the unique commitment-optimal partition recommends action $2$ so frequently that $\max B_2 > \gamma_3$.
In states strictly higher than $\gamma_3$, the sender is strictly better off by fully revealing the state, rendering the commitment solution not implementable.

\subsection{Sufficient Conditions for Implementability} \label{subsection:sufficient-cond}
In what follows, we identify conditions on model primitives that guarantee the existence of an implementable commitment solution. Under these conditions, the sender does not benefit from commitment relative to the sender-preferred equilibrium. The equilibrium payoff set and the sender-preferred equilibrium partition can thus be found by solving the corresponding commitment problem. Conversely, the commitment assumption is unnecessary for any information-design problem that satisfies these conditions.

To state the result, let $h(\gamma_i;\gamma_{i+1})$ denote the unique solution of $\mathbb{E}[\omega \mid \omega \in [h(\gamma_i;\gamma_{i+1}), \allowbreak \gamma_{i+1}]] = \gamma_i$ if it exists, and set it to $0$ otherwise. In words, $h(\gamma_i;\gamma_{i+1})$ is the state such that the conditional mean of the states between it and $\gamma_{i+1}$ is exactly $\gamma_i$.

\begin{proposition} \label{prop:sufficient}
	Suppose there are three or more actions. Then every commitment solution associated with a bi-pooling partition is implementable if
	\begin{equation} \label{eq:convexity}
		\frac{u_{i+1}-u_{i}}{\gamma_{i+1}-\gamma_{i}}>\frac{u_{i}-u_{i-1}}{\gamma_{i} - \max \left\{\gamma_{i-1},h\left(\gamma_{i} ; \gamma_{i+1}\right)\right\}}
	\end{equation}
	for all $i =2, \ldots, n-1$. Consequently, the commitment payoff is attained in an equilibrium of the disclosure game.
\end{proposition}

In Condition \eqref{eq:convexity}, $u_{i+1} - u_{i}$ is the sender's marginal benefit of inducing a higher action evaluated at action $i$ and $\gamma_{i+1} - \gamma_i$ is the difference in cutoffs for inducing actions $i+1$ and $i$ under complete information, respectively.
\cref{prop:sufficient} suggests that in a communication environment, if inducing a marginally higher action is either sufficiently more profitable or sufficiently more difficult (requiring a sufficiently larger expected state), or both, then the sender does not benefit from commitment power.
Put differently, the sender does not value commitment when his value function increases sufficiently fast in the expected state.

The intuition behind \cref{prop:sufficient} is as follows. 
To establish the sufficiency of \eqref{eq:convexity}, we argue that any partition that is not revelation-proof must also fail optimality in the commitment problem.
Take any barely obedient bi-pooling partition $\mathcal{B}$ that is not revelation-proof.
Then some action $i$ is recommended in some states in which the sender would prefer to fully reveal the state to induce a higher action instead; that is, $\max B_i > \gamma_{i+1}$.
To illustrate how $\mathcal{B}$ can then be strictly improved, suppose $B_i$ is an interval.
Then $\max B_i > \gamma_{i+1}$ and obedience implies $\gamma_{i+1} \in B_i$.
Next, modify the partition by shrinking $B_i$ and shifting the probability of recommending action $i$ to actions $i-1$ and $i+1$.
Recommending action $i+1$ can still be made barely obedient by inducing belief $\gamma_{i+1}$.
At the same time, action $i-1$ can now be induced at $\gamma_{i-1}$ if $\gamma_{i-1}$ is not too low, and otherwise at $h(\gamma_{i};\gamma_{i+1})$.
Either way, condition \eqref{eq:convexity} ensures that such a local mean-preserving spread is profitable by requiring the sender's utility to be ``convex enough'' with respect to the cutoffs $\{\gamma_i\}_{i\in N}$.
Consequently, any barely obedient partition associated with a commitment solution must be revelation-proof, and thus every such commitment solution is implementable.

Imposing a further assumption on the prior, Condition \eqref{eq:convexity} can be simplified. 

\begin{corollary} \label{cor:inc-dens}
    Suppose there are three or more actions. If $f$ is increasing, and $u_{i+1}-u_{i} \ge u_{i}-u_{i-1}$ and $\gamma_{i+1}-\gamma_{i} \le \gamma_{i}-\gamma_{i-1}$ hold for all $i = 2, \ldots, n-1$, with at least one inequality being strict for some $i$, then every commitment solution associated with a bi-pooling partition is implementable. Consequently, the commitment payoff is attained in an equilibrium.
\end{corollary}

An increasing prior density is equivalent to a convex prior CDF. This condition is satisfied by the uniform distribution and more generally by the family of power distributions on $[0,1]$, with the CDF given by $F(x) = x^\alpha$, $\alpha \ge 1$.

\subsection{The Special Case of Ternary Actions} \label{ss:ter}
It is instructive to take a deeper dive into the case in which the receiver has three actions. In this case, a partition can be written as $\mathcal{B} = \{B_1, B_2, B_3\}$. As implied by \cref{thm:sender-pref-eq-characterization}, in a sender-preferred laminar equilibrium partition, $B_1$ must be an interval if it is not null. Moreover, $B_2$ and $B_3$ either are both intervals or are such that $B_2 \subseteq \conv{(B_3)}$, meaning that $\mathcal{B}$ is a bi-pooling partition.\footnote{Both $B_1$ and $B_2$ may be empty, but $B_3$ cannot be: otherwise, revelation proofness is violated.} Consequently, \cref{thm:sender-pref-eq-characterization} implies that revelation proofness boils down to $\max B_1 \le \gamma_2$.

Armed with these observations, a sender-preferred equilibrium can be explicitly solved.

\begin{proposition} \label{claim:ternary-sp}
	Suppose that $n = 3$.
    A sender-preferred equilibrium is associated with a barely obedient bi-pooling partition $\mathcal{B}$ that either is also associated with the commitment solution or is such that $B_1 = [0, y]$, $B_2 = [z, \gamma_3]$, and $B_3 = [y, z] \cup [\gamma_3,1]$, where $z > 0$ and $y \ge 0$ solve the following system of equations:
	\[
	\begin{split}
	    \mathbb{E}\left[\omega \mid \omega \in\left[z, \gamma_{3}\right]\right] & = \gamma_{2}; \\
	    \mathbb{E}\left[\omega \mid \omega \in\left[y, z\right] \cup\left[\gamma_{3}, 1\right]\right]& =\gamma_{3}.
	\end{split}
	\]
\end{proposition}

When there are only three actions, the only reason that a commitment solution is not implementable is that the ``middle'' action, $2$, is recommended too often. Therefore, if no commitment solution is implementable, in the bi-pooling partition associated with a sender-preferred equilibrium, action $2$ is recommended as frequently as revelation proofness allows: that is, the upper bound of $B_2$ must coincide with $\gamma_3$.
 
The sufficient conditions can be further simplified when $|N| = 3$. 

\begin{corollary} \label{corollary:ternary-sufficient}
    If $n = 3$, $f$ is increasing, and $u_3 > 2u_2$, then all commitment solutions are implementable.
\end{corollary}

Applying \cref{corollary:ternary-sufficient} to our introductory example, since $u_2 = p$ and $u_3 = 1$, as long as $p < 0.5$, the seller does not benefit from commitment power. In other words, even prior to solving the commitment problem, we know that a sender-preferred equilibrium partition can be identified from the commitment solution. 

\section{Applications} \label{s:apl}
\subsection{Selling with Quality Disclosure}
We first consider a variant of the sales encounter model studied in Section 5 of \cite{milgrom1981good}. The state of the world, $\omega$, is interpreted as the quality of the seller's product. Let $p>0$ be the unit price, and for simplicity, assume that there is no quantity discount. Denote the seller's constant unit cost by $c$, where $0 \le c < p$. The product is indivisible: the buyer can only buy integer units of the product. The buyer's utility from purchasing $q$ units is $\omega U(q) - p q$, where $U\colon \mathbb{R}_+ \to \mathbb{R}$ is a bounded, strictly increasing, strictly concave three-times differentiable function with $U(0) = 0$. We further assume that $U(q) - pq$ is maximized at $n > 1$. As a consequence, the buyer buys at most $n$ units of the product, and she buys nothing if $\omega$ is close enough to 0.

The only significant difference between this model and that of \cite{milgrom1981good} is that he considers a perfectly divisible product and hence the seller's value function is strictly increasing. In our case, however, indivisibility makes the seller's value function a step function with $n$ jumps. For a perfectly divisible product, \citeauthor{milgrom1981good} shows that every equilibrium of the game features full revelation: the seller sends $m = \{\omega\}$ for each $\omega \in (0,1]$, resulting in the buyer-preferred outcome. With indivisibility, however, the seller may be able to gain considerably from verifiable communication, attaining his commitment payoff.

To state the result, let
\[
A^U(q) \coloneqq -\frac{U''(q)}{U'(q)} \quad \text{and} \quad P^U(q) \coloneqq -\frac{U'''(q)}{U''(q)}
\]
denote the coefficients of absolute risk aversion and absolute prudence, respectively.

\begin{proposition} \label{msm}
	If $f$ is increasing and $P^U > 2A^U$, there exists an equilibrium of this game in which the seller is as well off as he would be if he had commitment power.
\end{proposition}

The assumption of an increasing prior density can be interpreted as meaning that it is common knowledge that the consumer is relatively confident about the quality of the product. The condition $P^U > 2A^U$ is satisfied by, for example, a constant relative risk aversion (CRRA) utility function with parameter $0 < \sigma < 1$.

\subsection{Influencing a Voter}

Consider an amendment voting setting where a voter chooses among three alternatives:
maintaining the status quo (no bill, action $1$); the amended bill (action $2$); and
the (unamended) bill (action $3$).\footnote{See \cite{ek80} and \cite{e81}
for examples of amendment voting.} The state of the world is $\omega \in [0,1]$. The
voter has linear preferences: her utility from action $k \in \{1,2,3\}$ in state $\omega$ is given by 
$\alpha_k + \lambda_k \omega$. Moreover, $\lambda_3 > \lambda_2 > \lambda_1 = 0$
and $0 = \alpha_1 > \alpha_2 > \alpha_3$. Let $\gamma_2$ and $\gamma_3$ denote the cutoff
states at which the voter is indifferent between actions $1$ and $2$, and actions $2$
and $3$, respectively:\footnote{That is,%
\[
\gamma_2 = -\frac{\alpha_2}{\lambda_2} \quad \text{and} \quad \gamma_3 =
\frac{\alpha_2 - \alpha_3}{\lambda_3 - \lambda_2}.
\]
}
we impose $\gamma_2 < \gamma_3$. The voter's preferences are illustrated in \cref{vot}.

\begin{figure}[ht!]
    \centering
    \begin{tikzpicture}[scale=0.9, >=Stealth]
        \draw[thick, ->] (0,0) -- (6,0) node[below left]{$\omega$};
        \draw[thick, ->] (0,-2.5) -- (0,3.5) node[below right]{voter's utility};
        \draw[thick, Orange] (0,0) -- (5,0) node[above right]{action $1$};
        \draw[thick, blue] (0,-0.75) -- (5,1.5) node[right]{action $2$};
        \draw[thick, red] (0,-2) -- (5,2) node[above right]{action $3$};
        \draw (5/3,0) -- (5/3,0.1) node[above]{$\gamma_2$};
        \draw[dashed] (25/7,0) node[below]{$\gamma_3$} -- (25/7,0.8572);
    \end{tikzpicture}
    \caption{The voter's utilities.}
    \label{vot}
\end{figure}

In this model, the voter's preferences over action $k$ are characterized by two parameters: $\alpha_k$, her \emph{reference point} for action $k$ (her cardinal utility when the state is zero), and $\lambda_k$, her \emph{state sensitivity} for action $k$
(how fast her cardinal utility increases in the state). The voter agrees that when the state is low (intermediate, high), action $1$ (action $2$, action $3$, respectively) is optimal.

There is an expert who observes the state and discloses verifiable information to the voter \citep[e.g.,][]{jt13}. The expert's preferences satisfy $u_3 > u_2 > u_1 = 0$; that is, the expert strictly prefers the bill to the
amended bill, and the amended bill to no bill. This is a direct application of our model with $n = 3$ and cutoffs $\gamma_2$, $\gamma_3$, as defined above.\footnote{More generally, consider a group of voters with an odd number of members, each with linear preferences $\alpha^j_k + \lambda^j_k\omega$ satisfying the same sign restrictions as above. Suppose a voting rule satisfying the Condorcet winner criterion is employed. When voters' preferences are single-peaked, the Condorcet winner coincides with the median voter's most preferred alternative. If the median voter is the same across all states---which holds, for example, when voters are ordered so that $\gamma^i_2 \ge \gamma^j_2$ and $\gamma^i_3 > \gamma^j_3$ for all $i > j$---then the expert's problem reduces to communicating with the median voter alone, and all results below apply to the setting verbatim.} The expert-preferred equilibrium is therefore characterized by \cref{claim:ternary-sp}.

\cref{claim:influencing-voters} shows that the expert can be hurt if the voter becomes ``more inclined toward'' the bill in the sense that, all else equal, either the reference point or the state sensitivity for action $3$ increases.

\begin{proposition} \label{claim:influencing-voters}
    If no commitment solution is implementable and at least one of the following
    happens:
    \begin{itemize}[noitemsep,topsep=2.5pt]
        \item[(i)] the voter's state sensitivity for the bill, $\lambda_3$, increases,
        \item[(ii)] the voter's reference point for the bill, $\alpha_3$, increases,
    \end{itemize}
    then the expert's payoff in his preferred equilibrium may decrease.
\end{proposition}

When either (i) or (ii) occurs (or both), $\gamma_3$ decreases. This implies that the expected state required to pass the bill is lowered, which benefits the expert directly: he can more frequently induce passage of the original bill. However, this also introduces an adverse indirect effect. Recall from the discussion after \cref{claim:ternary-sp} that when no commitment solution is implementable, the upper
bound of $B_2$ is pinned to $\gamma_3$ by the binding revelation-proofness constraint. As $\gamma_3$ falls, this constraint tightens, limiting the range of states over which the expert can credibly recommend action $2$. The expert is therefore harmed when this indirect effect dominates the direct one. This may occur when the utility gap between the bill and the amended bill is smaller than that between the amended bill and the status quo (i.e., $u_3 - u_2 < u_2 - u_1$), which is plausible in many voting scenarios.

\section{Conclusion}

This paper revisits a classic persuasion game \citep[e.g.,][]{milgrom1981good} by changing one feature of the canonical model: the receiver has only finitely many actions. When the receiver can adjust her action finely, even a small improvement in beliefs leads to a better action for the sender, and the usual unraveling logic leaves little room for withholding information. With finitely many actions, by contrast, small changes in beliefs need not change the receiver's action. This creates room for credible pooling even when the sender can prove any true fact, so full revelation is no longer the sole relevant equilibrium outcome. We show that the resulting set of equilibrium payoffs is organized by partitions of the state space with a simple geometric property---the laminar structure. Any sender-worst equilibrium is payoff equivalent to a fully revealing equilibrium, whereas a sender-preferred equilibrium often involves pooling nonadjacent states. We characterize such sender-preferred equilibria and use the characterization to compare sender-preferred equilibrium payoff with the commitment payoff: the payoff the sender would obtain if he could commit to an information structure, as in the information-design literature. This comparison identifies conditions under which the commitment payoff can be achieved in equilibrium.

The analysis yields broader lessons about disclosure and information design. The same environment can sustain a fully revealing equilibrium, a sender-preferred equilibrium with substantial pooling, and every payoff in between through simple laminar equilibria. This observation helps explain the variety of disclosure policies observed in practice and qualifies the welfare conclusions one can draw from the existence of a fully revealing equilibrium alone: observing that voluntary disclosure admits a fully revealing equilibrium does not imply that it delivers the receiver-preferred outcome; mandatory disclosure can still substantially improve receiver welfare by eliminating sender-preferred pooling. At the same time, our sufficient conditions for attaining the commitment payoff identify information-design problems for which the commitment assumption can be relaxed.
On the other hand, if those conditions are met, the sender-preferred equilibrium and the equilibrium payoff set can be identified by solving the corresponding information-design problem.
When the commitment payoff is unattainable, the sender-preferred equilibrium captures the trade-off between inducing actions more favorable to the sender and ensuring that no type wants to deviate by fully revealing the state.

\appendix

\section{Omitted Proofs and Details} \label{section:appendix}
\subsection{Proofs for \autoref{section:eq-analysis}}

\subsubsection{Preliminary characterization}

In what follows, we employ the following result, which is essentially Theorem 1 in \cite{titova2025persuasion} adapted to our setting.\footnote{The messages in this paper are closed subsets of the state space, while in \cite{titova2025persuasion} they are Borel subsets of the state space.} We provide a proof below for completeness.

\begin{lemma}\label{lemma:eqm-partition}
    $\mathcal{B}$ is an equilibrium partition if and only if it is obedient and revelation-proof.
\end{lemma}
\begin{proof}
    Necessity is straightforward: if $\mathcal{B}$ is not revelation-proof or obedient, then in every associated assessment, the sender has a profitable deviation to full revelation or the receiver is not best responding to an on-path message. For sufficiency, suppose that $\mathcal{B}$ is an obedient and revelation-proof partition. Let $(\sigma,\tau,\beliefm)$ be an associated assessment such that the receiver (1) has {\it maximally skeptical} off-path beliefs and (2) best responds according to her posterior mean and breaks ties in the sender-adversarial manner when indifferent. Specifically, for each $m \notin \mathcal{B}$, let $\beliefm(\min m \sep m) = 1$, and $\tau(j \sep m) = 1$ if $\min m \in [\gamma_{j-1},\gamma_j)$ and $j \in N \smallsetminus \{ n \}$, or $\min m \in [\gamma_{n-1},\gamma_n]$ and $j = n$. Then, $(\sigma,\tau,\beliefm)$ is an equilibrium, and thus $\mathcal{B}$ is an equilibrium partition. Indeed, equilibrium conditions \ref{eq-cond-2}, \ref{eq-cond-3}, and \ref{eq-cond-4} are satisfied by construction. Equilibrium condition \ref{eq-cond-1} (the sender has no profitable deviations at each $\omega$) is satisfied because if $\omega \in B_i$ and $\omega \notin B_{i+1} \cups B_n$, then the sender's interim payoff is $u_i$; deviations to on-path messages with a higher index are not feasible, and deviations to on-path messages with a lower index or to off-path messages yield an interim payoff of at most $u_i$.
\end{proof}

\subsubsection{Towards the proof of \cref{thm:eq-payoff-set}}

To prove \cref{thm:eq-payoff-set}, we will need four auxiliary results.

\begin{claim}\label{claim:n-pt-supp}
    Let $H$ be a mean-preserving contraction of $F$ with $|\supp{(H)}| \le n$. Then there exist cutoffs  $0 \eqqcolon d_0 \le d_1 \le \ldots \le d_{m-1} \le d_m \coloneqq 1$ with $m \le |\supp{(H)}|$ such that $\int_{0}^{x} H(q) \ddd q \le \int_{0}^{x} F(q) \ddd q$ on $[d_{i-1}, d_{i}]$ for all $i = 1, \ldots, m$ and the inequality binds only at $d_{i-1}$ and $d_{i}$.
\end{claim}

\begin{proof}[Proof of \cref{claim:n-pt-supp}]
    First, let $\supp(H) = \{\widetilde{\gamma}_i\}_{i=1}^{k}$, where $k \le n$ since $|\supp(H)| \le n$ and $\widetilde{\gamma}_i < \widetilde{\gamma}_j$ if $i < j$. 
    Since $F(0) = H(0) = 0$, $F$ is strictly increasing (since $f > 0$), and $H$ is a step function, there exists $\varepsilon > 0$ such that $F(x) > H(x)$ for all $x \in [0, \varepsilon]$. Since $H \in MPC(F)$, we have $\int_{0}^{1} H(q) \ddd q = \int_{0}^{1} F(q) \ddd q$, and hence the set \[D_1 = \left\{x \in [\varepsilon,1] \colon \int_{0}^{x} H(q) \ddd q = \int_{0}^{x} F(q) \ddd q\right\}\] is nonempty. Let $d_1 = \inf D_1$. If $d_1 = 1$, then set $m=1$ and the proof is complete. For the rest of the proof, suppose that $d_1 < 1$. 
    
    Observe that $\widetilde{\gamma}_{1} < d_1 < \widetilde{\gamma}_{k}$: if $d_1 \le \widetilde{\gamma}_{1}$, then $F(x) > H(x) = 0$ for all $x \in [0, d_1)$, which contradicts the definition of $d_1$; if $d_1 \ge \widetilde{\gamma}_{k}$, then $F(x) < H(x) = 1$ for all $x \in [d_1, 1]$, so that $\int_{0}^{1} H(q) \ddd q \neq \int_{0}^{1} F(q) \ddd q$, a contradiction. 
    
    Next, we argue that $F(d_1) \ge H(d_1)$. Suppose to the contrary that $F(d_1) < H(d_1)$. Then, since $H$ is a CDF and hence right-continuous, there must exist $\delta > 0$ such that $F(x) < H(x) = H(d_1)$ for all $x \in [d_1, d_1+\delta]$, where the equality follows from the fact that $H$ is a step function. Then,
    \begin{align*}
        \int_{0}^{d_1+\delta} F(q) \ddd q & = \int_{0}^{d_1} F(q) \ddd q + \int_{d_1}^{d_1+\delta} F(q) \ddd q = \int_{0}^{d_1} H(q) \ddd q + \int_{d_1}^{d_1+\delta} F(q) \ddd q \\
        & < \int_{0}^{d_1} H(q) \ddd q + \int_{d_1}^{d_1+\delta} H(q) \ddd q,
    \end{align*}
    which contradicts the assumption that $H$ is a MPC of $F$. 
    
    Now, let $j = \min\{i:\widetilde{\gamma}_{i} > d_1\}$. Given that  $d_1 < \widetilde{\gamma}_{j}$, $F(d_1) \ge H(d_1)$, $F$ is strictly increasing, and $H$ is a step function, there exists $\eta > 0$ such that $F(x) > H(x)$ for all $x \in [d_1, d_1 + \eta]$. Consequently, $\int_{0}^{1} H(q) \ddd q = \int_{0}^{1} F(q) \ddd q$ implies that the set 
    \[D_2 = \left\{x \in [d_1+\eta,1] \colon \int_{0}^{x} H(q) \ddd q = \int_{0}^{x} F(q) \ddd q\right\}\] 
    is nonempty. Let $d_2 = \inf D_2$. If $d_2 = 1$, then set $m=2$ and the proof is complete. For the remainder of the proof, suppose that $d_2 < 1$. Using the same steps as above, one can show that $d_2 > \widetilde{\gamma}_{j}$ and $F(d_2) \ge H(d_2)$. Proceeding inductively, one can find $d_{t}$ with $t \le k \le n$ such that $d_t > \widetilde{\gamma}_{k}$. It must be that $d_t = 1$: suppose not, then because $F$ is strictly increasing, $F(x) < 1$ on $(d_k, 1)$; but $H(x)  = 1$ on the same interval, which implies that $\int_{0}^{1} H(q) \ddd q > \int_{0}^{1} F(q) \ddd q$, a contradiction. Now set $m = t$, the proof is complete.
\end{proof}

\begin{claim} \label{claim:induce-laminar}
    If $H$ is a mean-preserving contraction of $F$ with $|\supp(H)| \le n$, then it induces a laminar partition.
\end{claim}

\begin{proof}
    By \cref{claim:n-pt-supp}, on each of the $m$ intervals such that the MPC constraint only binds at the endpoints, the mass is redistributed to at most $n$ points, and there can be at most $m$ such intervals. We show that every such interval admits a laminar partition; the definition of a laminar partition then implies that the resulting partition is still laminar by taking the union.
        
    Denote an arbitrary interval on which the MPC constraint only binds at the endpoints by $I := [a,b]$; that is, $\int_{0}^{x} H(q) \ddd q \le \int_{0}^{x} F(q) \ddd q$ for all $q \in I$, and the inequality binds only at $a$ and $b$.
    
    The remainder of the proof is very similar to the proof of Lemma 11 in \cite{candogan2023optimal}, and hence we only provide an outline here; readers interested in details are directed to that paper. Let $K = |\supp{(H)} \cap I| \le n$; the proof proceeds by induction on $K$. If $K = 1$, let $\supp{(H)} \cap I = \{\widetilde{\gamma}_\ell\}$; then clearly $\{B_\ell\}$ where $B_\ell = I$ is a laminar partition of $I$. If $K = 2$, let $\supp{(H)} \cap I = \{\widetilde{\gamma}_\ell, \widetilde{\gamma}_m\}$ where $\ell < m$; then by Lemma 4 in \cite{absy}, $\{B_\ell, B_m\}$ can be chosen such that $B_\ell = [c,d]$ and $B_m = [a,c] \cup [b,d]$, which is laminar. 
    
    Taking $K = 2$ as the base case, consider $K > 2$; the induction hypothesis holds for $K - 1$. One can find a closed interval $B_\ell$ such that (i) $\mu_F(B_\ell) = h(\widetilde{\gamma}_\ell)$, where $h$ is the probability mass function (pmf) of $H$, and $\widetilde{\gamma}_\ell = \min (\supp{(H)} \cap I)$,\footnote{The only difference between our proof and \cite{candogan2023optimal}'s is that they choose the mass point of $H$ that has the largest index number on $I$, and for our purpose we work with the smallest. Their proof, however, goes through despite this difference.} and (ii) $\mathbb{E}[\omega \mid \omega \in B_\ell] = \widetilde{\gamma}_\ell$. Consequently, conditional on $\omega \not \in B_\ell$, $H$ only has $K-1$ mass points, and Lemma 12 in \cite{candogan2023optimal} shows that it is a MPC of $F$. Invoking the inductive hypothesis, a laminar partition of $I$, $\{\hat{B}_i\}_{i \in \mathcal{T}}$, is obtained, where $\mathcal{T}:=\{k \ne \ell: \gamma_k \in \supp{(H)} \cap I\}$. For every $i \in \mathcal{T}$, let $B_i \coloneqq \clos{(\hat{B}_i \smallsetminus B_\ell)}$. Since $\{\hat{B}_i\}_{i \in \mathcal{T}}$ is laminar, and $\ell < i$ for all $i \in \mathcal{T}$, $\{B_i\}_{i \in \mathcal{T} \cup \{\ell\}}$ is also laminar. 
\end{proof}

\begin{claim}\label{claim:interval-good}
    Suppose that $W \subseteq [0,1]$ and $B = [x,y] \subseteq [0,1]$ such that $\mu_F(W) = \mu_F(B)$ and $\mathbbm{E}[\omega \sep \omega \in W] = \mathbbm{E}[\omega \sep \omega \in B]$. Then, $\inf W \leq x$ and $y \leq \sup W$.
\end{claim}

\begin{proof}
    We prove that $y \leq \sup W$; the proof of $\inf W \leq x$ is analogous. Since $\mu_F(W) = \mu_F(B)$ and $\mathbbm{E}[\omega \sep \omega \in W] = \mathbbm{E}[\omega \sep \omega \in B]$, we have $\int_{W \smallsetminus B} \omega \ddd F(\omega) = \int_{B \smallsetminus W} \omega \ddd F(\omega)$ and $\mu_F(W \smallsetminus B) = \mu_F(B \smallsetminus W)$.

    Suppose, by contradiction, that $\sup W < y$. Then, $W \smallsetminus B \subseteq [0,x]$, $B \smallsetminus W \subseteq [x,y]$ and $\mu_F(W\smallsetminus B) = \mu_F(B \smallsetminus W) > 0$. Furthermore,
    \begin{align*}
        \int_{W\smallsetminus B} \omega \ddd F(\omega) < x \mu_F(W\smallsetminus B) = x \mu_F(B \smallsetminus W)
        < \int_{B \smallsetminus W} \omega \ddd F(\omega),
    \end{align*}
    a contradiction. Therefore, $y \leq \sup W$.
\end{proof}

\begin{claim}\label{claim:tiebreaking}
    Let $(\sigma',\tau',\beliefm')$ be an equilibrium. Define $\tau''$ by
\[
\tau''(i\sep m)=
\begin{cases}
1, & \text{if } i=\max \supp\tau'(\cdot\sep m)\\
0, & \text{otherwise}
\end{cases}
\]
for all $m \in \mathcal{C}$ and $i \in N$. Then $(\sigma',\tau'',\beliefm')$ is also an equilibrium.
\end{claim}

\begin{proof}
Equilibrium conditions 3 and 4 are unchanged because the sender's strategy and the belief system are identical to those in the original equilibrium. Condition 2 also remains true: if $\tau''(i\sep m)>0$, then
$i\in \supp\tau'(\cdot\sep m)$; since $(\sigma',\tau',\beliefm')$ is an equilibrium, \((\sigma^{\prime}, \tau^{\prime \prime}, \beliefm^{\prime})\) must also satisfy condition 2. 

It remains to verify condition 1. To simplify notation, for each message $m$, let
\[
w'(m):=\sum_{i\in N}\tau'(i\mid m)u_i,
\quad \text{and} \quad
w''(m):=\sum_{i\in N}\tau''(i\mid m)u_i.
\]
Because the receiver is only indifferent between actions $i$ and $i+1$ when the expected state after seeing a message is exactly $\gamma_{i+1}$, for every message $m$, $\supp\tau'(\cdot\sep m)$ is either a singleton $\{i\}$ or a pair of adjacent actions $\{i,i+1\}$.

Define a function $c: [0, u_n] \to [0,u_n]$ by $c(z):=u_{\min\{k\in N:\ z\le u_k\}}$. Because $u_i$ is strictly increasing in $i$, $c$ is increasing. For every message $m$ we have $w''(m)=c(w'(m))$: indeed, if $\supp\tau'(\cdot\sep m)=\{i\}$, then $w'(m)=u_i$ and $w''(m)=u_i=c(w'(m))$; if $\supp\tau'(\cdot\sep m)=\{i,i+1\}$, then $w'(m)\in (u_i,u_{i+1})$ and $w''(m)=u_{i+1}=c(w'(m))$.

Now fix $\omega\in\Omega$. Since $(\sigma',\tau',\beliefm')$ is an equilibrium,
\[
\supp\sigma'(\omega)\subseteq
\operatorname*{arg\,max}_{m\in M(\omega)} w'(m).
\]
Because $c$ is increasing,
\[
\operatorname*{arg\,max}_{m\in M(\omega)} w'(m)
\subseteq
\operatorname*{arg\,max}_{m\in M(\omega)} c(w'(m))
=
\operatorname*{arg\,max}_{m\in M(\omega)} w''(m).
\]
Therefore
\[
\supp\sigma'(\omega)\subseteq
\operatorname*{arg\,max}_{m\in M(\omega)} w''(m),
\]
so condition 1 also holds for $(\sigma',\tau'',\beliefm')$. Thus, $(\sigma',\tau'',\beliefm')$ is an equilibrium.
\end{proof}

\subsubsection{Proof of \autoref{thm:eq-payoff-set}} \label{prf:eq-payoff-set}

\paragraph{Proof of Part \ref{thm1-sw}.}

First, observe that the fully revealing partition $\mathcal{A}$ is obedient and revelation-proof since $A_i = [\gamma_{i}, \gamma_{i+1}]$ for each $i \in N$. By \cref{lemma:eqm-partition}, it is an equilibrium partition. Next, we show that the sender's ex-ante payoff cannot be lower than $\underline{V}$ in any other equilibrium. Let $\underline{w}(\omega)\coloneqq \min\{u_i: i \in N, \, \omega \in A_i\}$ be the sender's payoff in state $\omega$ when the receiver knows the state and breaks ties in the sender-adversarial manner. By definition, $\int_{\Omega} \underline{w}(\omega) \ddd F(\omega) = \sum_{i \in N} \mu_F(A_i)u_i = \underline{V}$. If the sender's ex-ante payoff is strictly below $\underline{V}$ in an equilibrium, then his interim payoff is strictly below $\underline{w}(\omega)$ in a positive measure of states. Then, in each of those states, the sender has a profitable deviation toward sending message $\{ \omega \}$ and receiving $\underline{w}(\omega)$. Therefore, the sender's ex-ante payoff in a sender-worst equilibrium is exactly $\underline{V}$. This completes the proof of Part \ref{thm1-sw}.

\paragraph{Proof of Part \ref{thm1-sp}.}

We proceed in three steps. First, we show that if there exists a sender-preferred equilibrium, then there exists a sender-preferred laminar partitional equilibrium (\cref{lemma:sender-preferred-laminar}). Second, we show that a sender-preferred laminar partitional equilibrium exists (\autoref{lemma:payoff-max-laminar}). Third, we show that $\overline{V} > \underline{V}$ (\cref{lemma:exante-exceed}).

\begin{lemma}\label{lemma:sender-preferred-laminar}
    For any equilibrium in which the receiver plays pure strategy, there exists a laminar equilibrium partition that induces the same posterior mean distribution.
    Furthermore, if there exists a sender-preferred equilibrium, then there exists a sender-preferred laminar partitional equilibrium.
\end{lemma}
\begin{proof}[Proof of \cref{lemma:sender-preferred-laminar}]

Let $(\sigma, \tau, \beliefm)$ be an equilibrium in which the receiver plays pure strategies. For any state $\omega \in \Omega$ in which the sender mixes, i.e., $|\supp \, \sigma(\cdot \sep \omega)| > 1$, since $(\sigma, \tau, \beliefm)$ is an equilibrium, there exists $i \in N$ such that $\tau(i \sep m) = \tau(i \sep m') = 1$ for every $m, m' \in \supp \, \sigma(\cdot \sep \omega)$. Thus, in every $\omega \in \Omega$, there is an action $i$ played with probability 1 in this equilibrium. 
    For every $i \in N$, let 
    \[
    W_i \coloneqq \left\{\omega \in \Omega \colon \tau(i \sep m) = 1 \text{ for all } m \in \supp \, \sigma(\cdot \sep \omega)\right\}.
    \]
	By construction, $\cup_{i \in N} W_i = \Omega$, and $W_i \cap W_j = \varnothing$ for any $i \ne j$.
    By Theorem 1(a) of \cite{titova2025persuasion}, for every $i \in N$, $\widetilde{\gamma}_i \coloneqq \mathbb{E}[\omega \sep \omega \in W_i] \in [\gamma_{i}, \gamma_{i+1}]$, and $W_i \subseteq [0,\gamma_{i+1}]$. 

Consider the PMD $\widetilde{G}$ with $\supp \widetilde{G} \subseteq \{\widetilde{\gamma}_i\}_{i=1}^n$ whose probability mass function is given by $\widetilde{g}(\widetilde{\gamma}_i) = \mu_F(W_i)$. By construction, $\widetilde{G}$ is a mean-preserving contraction of the PMD induced by the equilibrium $(\sigma, \tau, \beliefm)$. Since the PMD induced by an equilibrium is itself a mean-preserving contraction of $F$, $\widetilde{G}$ is a mean-preserving contraction of $F$ with $|\supp(\widetilde{G})| \le n$. Then by \autoref{claim:induce-laminar}, $\widetilde{G}$ is also induced by a laminar partition $\mathcal{B}$ with generic element $B_i$. By construction, for every $i \in N$, $\mathbb{E}[\omega \sep \omega \in B_i] = \widetilde{\gamma}_i$, which implies obedience. Also by construction, $\mu_F(B_i) = \mu_F(W_i)$ for each $i \in N$.

We show next that $\mathcal{B}$ is also revelation-proof, i.e., $B_i \subseteq [0,\gamma_{i+1}]$ for all $i \in N$. Fix $i\in N$. Because $\mathcal{B}$ is laminar, there exists $k \le i$ such that for every action $j \in \{k, \ldots, i\}$, $\conv(B_j) \subseteq \conv(B_i)$; and for every action $j \in \{1, \ldots, k-1, i+1, \ldots, n\}$, $\mu_F(B_j \cap \conv{(B_i)}) = 0$.
This implies that $\widehat{B}_i \coloneqq B_k \cups B_i = \conv(B_i)$ is an interval. Let $\widehat{W}_i \coloneqq W_k \cups W_i$. Since for each $j=k,\ldots, i$, $\mu_F(B_j) = \mu_F(W_j)$ and $\mathbbm{E}[{\omega \sep \omega \in B_j}] = \mathbbm{E}[\omega \sep \omega \in W_j]$, and since $W_j \cap W_\ell = \varnothing$ and $\mu_F(B_j \cap B_\ell) = 0$ for any $j \ne \ell$, we have $\mu_F(\widehat{B}_i) = \mu_F(\widehat{W}_i)$ and $\mathbbm{E}[\omega \sep \omega \in \widehat{B}_i] = \mathbbm{E}[\omega \sep \omega \in \widehat{W_i}]$. 
By \cref{claim:interval-good}, we have $\widehat{B}_i \subseteq  \conv\left(\widehat{W_i}\right) \subseteq [0,\gamma_{i+1}]$, where the last inclusion follows since $W_i \subseteq [0,\gamma_{i+1}]$ for all $i \in N$.

To prove the second statement, suppose that there exists a sender-preferred equilibrium. An immediate consequence of \cref{claim:tiebreaking} is that in a sender-preferred equilibrium, the receiver mixes on a null set of states. Therefore, we can focus on a sender-preferred equilibrium $(\sigma,\tau,\beliefm)$ in which the receiver breaks ties in favor of the sender. 
Then by the first statement, there must exist a sender-preferred laminar equilibrium partition.
\end{proof}

\begin{lemma} \label{lemma:payoff-max-laminar}
    Among all obedient and revelation-proof laminar partitions, there is one that maximizes the sender's ex-ante payoff.
\end{lemma}

\begin{proof}
    The problem of finding a laminar partitional equilibrium that maximizes the sender's ex-ante payoff can be written as
\begin{align*}
    \max_{\mathcal{B} ~ \text{laminar}} ~ & \sum_{i \in N} u_i \mu_F(B_i) \\
    \text{s.t. } ~~ & \bigcup_{i \in N} B_i = [0,1] \\
    & \mu_F(B_i \cap B_j) = \varnothing \text{ for all } i \ne j \\
    & \gamma_i \le \mathbb{E}[\omega \mid \omega \in B_{i}] \le \gamma_{i+1} \text{ for all } i \in N \\
    & B_i \subseteq [0,\gamma_{i+1}] \text{ for all } i \in N
\end{align*}
To prove the lemma, it suffices to show that this problem has a solution, by \cref{lemma:sender-preferred-laminar}.

Although an action $j$ may be never recommended, one can still assume that $B_j$ is nonempty by setting $\mu_F(B_j) = 0$. Because $\mathcal{B}$ is laminar, it is without loss of generality to assume that for each $i \in N$, $B_i$ is the union of at most $n$ intervals (cf. \cref{obs:laminar-i-intervals}). Consequently, adding singletons if necessary, one can always set $B_i$ as the union of exactly $n$ convex sets $\{B_{i,s}\}_{s=1}^{n}$ such that $\mu_F(B_{i,s'} \cap B_{i, s''}) = 0$ for all $s' \ne s''$.

Let $\mathcal{C}_c([0,1])$ denote the set of closed, nonempty, and convex subsets of $[0,1]$ endowed with the Hausdorff distance; to simplify notation, we write $\mathcal{C}_c$ henceforth. By Proposition 1 in \cite{ely2022cake}, $\mathcal{C}_c$ is compact; by Tychonoff's theorem, $\mathcal{C}_c^{n^2}$ is compact in the product topology. The problem above can be transformed to
\begin{align}
    \max_{\{B_{i,s}\} \in \mathcal{C}_c^{n^2}} ~ & \sum_{i \in N} u_i \sum_{s=1}^{n} \mu_F(B_{i,s}) \label{tp1} \\
    \text{s.t. } \quad ~~ & \cup_{i,s} B_{i,s} = [0,1] \nonumber \\
    & \mu_F(B_{i',s'} \cap B_{i'', s''}) = 0 \text{ for all } (i',s') \ne (i'',s'') \nonumber \\
    & \mu_F\left(\cup_{s=1}^{n} B_{i,s}\right)\gamma_i \le \int_{\cup_{s=1}^{n} B_{i,s}} \omega \ddd  \mu_F(\omega) \le \mu_F\left(\cup_{s=1}^{n} B_{i,s}\right)\gamma_{i+1} \nonumber \\
    & \cup_{s =1}^{n} B_{i, s} \subseteq [0,\gamma_{i+1}]\nonumber
\end{align}
where the third constraint is equivalent to the conditional mean condition.

Define
\[
\mathcal{D} = \left\{\{B_{i,s}\} \in \mathcal{C}_c^{n^2} \colon \cup_{i,s} B_{i,s} = [0,1], \text{ and } \mu_F(B_{i',s'} \cap B_{i'', s''}) = 0 \text{ for all } (i',s') \ne (i'',s'')\right\};
\]
We claim that $\mathcal{D}$ is compact. To show this, it is enough to show that $\mathcal{D}$ is a closed subset of $\mathcal{C}_c^{n^2}$. Take any $\{B_{i, s}^m\}$ that converges to $\{B_{i, s}\}$ in the product topology, then $B^m_{i,s} \to B_{i,s}$ for each $i$ and $s$. Consequently, because the limit of convergence in Hausdorff distance is preserved under finite unions,\footnote{See, for example, Lemma 4.5 in \cite{mclennan2018advanced}.} $\cup_{i,s} B^m_{i,s} \to \cup_{i,s} B_{i,s}$. Therefore, if $\cup_{i,s} B^m_{i,s} = [0,1]$, it must be that $\cup_{i,s} B_{i,s} = [0,1]$. Furthermore, if $\mu_F(B^m_{i',s'} \cap B^m_{i'',s''}) = 0$ for all $m$ and $(i',s') \ne (i'',s'')$, the same argument as the second paragraph in the proof of Lemma 2 in \cite{ely2022cake} shows that $\mu_F(B_{i',s'} \cap B_{i'',s''}) = 0$ for all $(i',s') \ne (i'',s'')$. Therefore, if $\{B^m_{i,s}\} \in \mathcal{D}$ for each $m$ and $\{B_{i, s}^m\} \to \{B_{i, s}\}$, it must be that $\{B_{i, s}\} \in \mathcal{D}$. Thus, $\mathcal{D}$ is a closed subset of $\mathcal{C}_c^{n^2}$. 

Problem (\ref{tp1}) is equivalent to
\begin{align}
    \max_{\{B_{i,s}\} \in \mathcal{D}} ~ & \sum_{i \in N} u_i \sum_{s=1}^{n} \mu_F(B_{i,s}) \label{tp2} \\
    \text{s.t. } \quad ~~ & \left(\sum_{s=1}^{n} \mu_F(B_{i,s})\right)\gamma_i \le \sum_{s=1}^{n}\int_{B_{i,s}} \omega \ddd \mu_F(\omega) \le \left(\sum_{s=1}^{n} \mu_F(B_{i,s})\right) \gamma_{i+1} \label{eq:exi-ob} \\
    & \cup_{s =1}^{n} B_{i, s} \subseteq [0,\gamma_{i+1}] \label{eq:exi-rp}
\end{align}
where constraint (\ref{eq:exi-ob}) supersedes the third constraint in problem (\ref{tp1}) because for any $\{B_{i,s}\} \in \mathcal{D}$, $\mu_F(B_{i',s'} \cap B_{i'', s''}) = 0$ for all $(i',s') \ne (i'',s'')$. 

By the extreme value theorem, to show that a solution to problem (\ref{tp2}) exists, it suffices to show that (i) the objective function is continuous, and (ii) the constraint set is nonempty and compact. Clearly, the constraint set is nonempty: for each $i \in N$, consider
\[
B_{i,1} = A_i, \text{ and } B_{i,s} = \{\gamma_{i+1}\} \text{ for all } s = 2, \ldots, n,
\]
then $\{B_{i,s}\}$ is feasible for this problem. Furthermore, by Proposition 1 in \cite{ely2022cake}, $\mu_F$ is continuous on $\mathcal{C}_c$, and hence the objective is continuous. 

Since $\mathcal{D}$ is compact, to show that the constraint set is compact, it suffices to show that each of the constraints, (\ref{eq:exi-ob}) and (\ref{eq:exi-rp}), defines a closed subset of $\mathcal{D}$. Observe that if $C^m \to C$ and $D^m \to D$ with $C^m \subseteq D^m$ for all $m$, then $C \subseteq D$.\footnote{See, for example, page 15 in \cite{ely2022cake}.} Then because the limit of convergence in Hausdorff distance is preserved under finite unions, if $\cup_{s=1}^{n} B^m_{i,s} \subseteq [0,\gamma_{i+1}]$ for each $i$, it must be that $\cup_{s=1}^{n} B_{i,s} \subseteq [0,\gamma_{i+1}]$ for each $i$. Hence, (\ref{eq:exi-rp}) defines a closed subset of $\mathcal{D}$.

Next, we show that (\ref{eq:exi-ob}) does the same, which is equivalent to showing that if $\{B_{i, s}^m\} \to \{B_{i,s}\}$ where $\{B_{i, s}^m\} \in \mathcal{D}$ for each $m$, then 
\[
\left(\sum_{s=1}^{n} \mu_F(B^m_{i,s})\right)\gamma_i \le \sum_{s=1}^{n}\int_{B^m_{i,s}} \omega \ddd \mu_F(\omega) \le \left(\sum_{s=1}^{n} \mu_F(B^m_{i,s})\right) \gamma_{i+1}
\]
for all $m$ and $i$ implies that
\[
\left(\sum_{s=1}^{n} \mu_F(B_{i,s})\right)\gamma_i \le \sum_{s=1}^{n}\int_{B_{i,s}} \omega \ddd \mu_F(\omega) \le \left(\sum_{s=1}^{n} \mu_F(B_{i,s})\right) \gamma_{i+1}
\]
for all $i$. Because $\mu_F$ is continuous on $\mathcal{C}_c$, $\sum_{s=1}^{n} \mu_F(B^m_{i,s}) \to \sum_{s=1}^{n} \mu_F(B_{i,s})$ for all $i$. Consequently, $\left(\sum_{s=1}^{n} \mu_F(B^m_{i,s})\right) \gamma_i \to \left(\sum_{s=1}^{n} \mu_F(B_{i,s})\right) \gamma_i$, and $\left(\sum_{s=1}^{n} \mu_F(B^m_{i,s})\right) \gamma_{i+1} \to \left(\sum_{s=1}^{n} \mu_F(B_{i,s})\right) \gamma_{i+1}$ for each $i$. Therefore, it only remains to show that
\[
\sum_{s=1}^{n}\int_{B^m_{i,s}} \omega \ddd \mu_F(\omega) \to \sum_{s=1}^{n}\int_{B_{i,s}} \omega \ddd \mu_F(\omega),
\]
which is a consequence of $E \mapsto \int_E \omega \ddd  \mu_F(\omega)$ being continuous on $\mathcal{C}_c$. 

To prove this claim, we show that $E\mapsto\int_E \omega \ddd  \mu_F(\omega)$ is both upper- and lower-semicontinuous. To see that it is upper-semicontinuous, pick any $\varepsilon > 0$, and let $E \in \mathcal{C}_c$; we show that there exists $\delta > 0$ such that for every $E' \in N(E, \delta)$, where $N(E,\delta)$ is the $\delta$-neighborhood of $E$, $\int_{E'}\omega \ddd  \mu_F(\omega) < \int_{E} \omega \ddd  \mu_F(\omega) + \varepsilon$. Because $E \in \mathcal{C}_c$, there exist $a,b \in [0,1]$ such that $E = [a,b]$. The key observation here is that for any $E' \in N(E, \delta)$, it must be that $E' \subseteq [a-\delta, b+\delta]$. Then
\begin{align*}
    \int_{a-\delta}^{b+\delta} \omega \ddd  \mu_F(\omega) & = \int_{a-\delta}^{a} \omega \ddd  \mu_F(\omega) + \int_{a}^{b} \omega \ddd  \mu_F(\omega) + \int_{b}^{b+\delta} \omega \ddd  \mu_F(\omega) \\
    & \le \mu_{F}([a-\delta, a]) + \int_E \omega \ddd  \mu_F(\omega) + \mu_{F}([b,b+\delta])
\end{align*}
where the inequality holds because $\omega \in [0,1]$. For $\delta$ small enough, since $\mu_F$ is absolutely continuous with respect to the Lebesgue measure, \[\int_{E'} \omega \ddd  \mu_F(\omega) \le \int_{a-\delta}^{b+\delta} \omega \ddd  \mu_F(\omega) < \int_{E} \omega \ddd  \mu_F(\omega) + \varepsilon.\]

Next, we show that $\int_E \omega \ddd  \mu_F(\omega)$ is lower semicontinuous, that is, there exists $\delta > 0$ such that for every $E' \in N(E, \delta)$, $\int_{E'} \omega \ddd  \mu_F(\omega) > \int_{E} \omega \ddd  \mu_F(\omega) - \varepsilon$. Without loss of generality, assume $\delta < (b-a)/2$. Consequently, $[a+\delta, b-\delta]$ is an interval of positive measure, and for any $E' \in N(E, \delta)$, $[a+\delta, b-\delta] \subseteq E'$. Then
\begin{align*}
    \int_{a+\delta}^{b-\delta} \omega \ddd  \mu_F(\omega) & = \int_{a}^{b} \omega \ddd  \mu_F(\omega) - \int_{a}^{a+\delta} \omega \ddd  \mu_F(\omega) - \int_{b-\delta}^{b} \omega \ddd  \mu_F(\omega) \\
    & \ge \int_E \omega \ddd  \mu_F(\omega) - \mu_{F}([a, a+\delta]) - \mu_{F}([b-\delta,b])
\end{align*}
where the inequality follows from the fact that $\omega \ge a$ on $[a,b]$. Consequently, $\int_E \omega \ddd  \mu_F(\omega)$ is both upper- and lower-semicontinuous, and hence continuous. This completes the proof. 
\end{proof}

Let $\mathcal{B}$ denote a laminar partition associated with a sender-preferred equilibrium; by the previous two lemmas, such a partition exists. Denote the sender's ex-ante payoff from $\mathcal{B}$ by $\overline{V}$.

\begin{lemma} \label{lemma:exante-exceed}
    $\overline{V} > \underline{V}$.
\end{lemma}

\begin{proof}

Let $\varepsilon \in (0,\gamma_2)$. 
Also, for each $i \in N \smallsetminus \{ 1,n \}$, let $B_i^\varepsilon = A_i$, $B_1^\varepsilon = [\varepsilon,\gamma_2]$ and $B_n^\varepsilon = [0,\varepsilon] \cup A_n$. It is easy to see that $\mathcal{B}(\varepsilon) :=  \{ B_i^\varepsilon \}_{i \in N}$ is a revelation-proof partition. Furthermore, $\mathbbm{E}[\omega \sep \omega \in B_i^\varepsilon] \in A_i$ for all $i \in N \smallsetminus \{ n \}$ by construction. Finally, the function $\psi(\varepsilon) := \mathbbm{E}[\omega \sep \omega \in B_n^\varepsilon]$ 
is continuous at $0$ and $\psi(0) > \gamma_{n}$. Consequently, there exists $\varepsilon'> 0$ sufficiently small such that $\psi(\varepsilon') \in [\gamma_{n},1]$, which makes $\mathcal{B}({\varepsilon'})$ an obedient partition. By \cref{lemma:eqm-partition}, $\mathcal{B}({\varepsilon'})$ is an equilibrium partition; the sender's ex-ante payoff in that equilibrium is
\begin{align*}
    V(\mathcal{B}(\varepsilon')) = \underline{V} + \mu_F([0,\varepsilon']) (u_n - u_1) > \underline{V},
\end{align*}
which completes the proof.
\end{proof}

\paragraph{Proof of Part \ref{thm1-range}.}

For any $V \in [\underline{V}, \overline{V}]$, we construct an obedient and revelation-proof laminar partition that yields $V$ by ``blending'' the fully informative partition $\mathcal{A}$ and a sender-preferred laminar equilibrium partition $\mathcal{B}$ (which exists by part \ref{thm1-sp}).
For all $z \in \left[0, \gamma_n\right]$, 
\begin{itemize}
    \item for every $i \in N$ such that $\gamma_{i+1} \le z$, let 
$\hat{B}_{i}^{z} \coloneqq A_i$; 
    \item if $i$ is such that such that $\gamma_{i} \le z \le \gamma_{i+1}$, $\hat{B}_{i}^{z} \coloneqq \clos(B_i \smallsetminus [0,\gamma_i]) \cup [\gamma_i, z]$;
    \item for every $i \in N$ such that $z < \gamma_{i}$, let $\hat{B}_{i}^{z} \coloneqq \clos{(B_{i} \smallsetminus [0, z])}$.
\end{itemize}
For any $z \in\left[0, \gamma_n\right]$, and any $i \in N$, $\mathbb{E}\left[\omega \, \big| \, \omega \in B_{i}\right] \in A_i$ implies $\mathbb{E}\left[\omega \, \big| \, \omega \in \hat{B}_{i}^{z}\right] \in A_i$. Consequently, $\hat{\mathcal{B}}(z)\coloneqq\{\hat{B}_{i}^{z}\}_{i=1}^{n}$ is obedient. It is also revelation-proof because both $\mathcal{A}$ and $\mathcal{B}$ are. Finally, because both $\mathcal{A}$ and $\mathcal{B}$ are laminar, it can be checked using \eqref{eqn:laminar-definition} that $\hat{\mathcal{B}}(z)$ is a laminar partition. 

Let $V(z)$ denote the sender's ex-ante payoff from $\hat{\mathcal{B}}(z)$:
$
V(z) \coloneqq \sum_{i \in N} u_{i} \, \mu_{F}\left(\hat{B}_{i}^{z}\right);
$
because $f > 0$, $V(z)$ is continuous in $z$. Since $V(0)=\overline{V} \ge V \ge \underline{V} = V(\gamma_{n})$, by the intermediate value theorem, there exists $z^V \in [0, \gamma_n]$ such that $V(z^V) = V$. Thus, the partition $\hat{\mathcal{B}}(z^V)$ is an obedient and revelation-proof laminar partition that yields $V$.

\subsubsection{Proof of \autoref{thm:sender-pref-eq-characterization}}
To establish \autoref{thm:sender-pref-eq-characterization}, we first prove the following two preliminary results. 

\begin{claim} \label{claim:lowest-element}
    Let $\mathcal{B}$ be a laminar partition associated with a sender-preferred equilibrium, and let $j = \min\{i : \mu_F(B_i) > 0\}$. If $\mathbb{E}[\omega \, | \, \omega \in B_j] > \gamma_j$, then $B_j = [0,d]$ for some $d \le \gamma_{j+1}$. 
\end{claim}

\begin{proof}
    Because $\mathcal{B}$ is a laminar partition, $B_j$ must be an interval, and hence one can write $B_j = [\underline{b}_j, \overline{b}_j]$. Furthermore, revelation proofness implies that $[\underline{b}_j, \overline{b}_j] \subseteq [0, \gamma_{j+1}]$. 
    
    To show that $\underline{b}_j = 0$, it suffices to show that $\inter{(B_j)} \cap \conv{(B_k)} = \varnothing$ for all $k > j$.\footnote{For a subset $S$ of $[0,1]$, $\inter{(S)}$ denotes the interior of $S$.} Suppose not, so $B_j \subseteq \conv{(B_k)}$ for some $k > j$. The laminar structure implies that $B_k$ is the union of at most $k$ closed intervals; let $[\alpha, \beta]$ be an interval such that $\beta \le \underline{b}_j$ and $[\alpha,\beta] \subseteq B_k$; this interval is well-defined because $\mathcal{B}$ is laminar and $B_j \subseteq \conv{(B_k)}$. 
    Because $\mathbb{E}\left[\omega \, | \, \omega \in B_k \right] = \gamma_k$ and $\mathbb{E}\left[\omega \, | \, \omega \in B_j \right] := \gamma_S > \gamma_j$, for small enough $\varepsilon > 0$, one can find $h(\varepsilon)$ such that
    \[
    \mathbb{E}\left[\omega \, \big| \, \omega \in \left(B_k \cup \left[\overline{b}_j - h(\varepsilon), \overline{b}_j\right]\right) \smallsetminus (\beta-\varepsilon,\beta]\right] = \gamma_k,\]
    and 
    \[\mathbb{E}\left[\omega \, \bigg| \, \omega \in [\beta-\varepsilon,\beta] \cup \left[\underline{b}_j, \overline{b}_j-h(\epsilon)\right] \right] \ge \gamma_j.
    \]
    Now define a new partition $\widetilde{\mathcal{B}}$ with generic element $\widetilde{B}_i$ by $\widetilde{B}_i = B_i$ if $i \ne j,k$, $\widetilde{B}_j = [\beta-\varepsilon,\beta] \cup \left[\underline{b}_j, \overline{b}_j-h(\epsilon)\right]$, and $\widetilde{B}_k = \left(B_k \cup \left[\overline{b}_j - h(\varepsilon), \overline{b}_j\right]\right) \smallsetminus (\beta-\varepsilon,\beta]$. $\widetilde{\mathcal{B}}$ is obedient by construction, and it is also revelation-proof because $\mathcal{B}$ is; then by \autoref{lemma:eqm-partition}, $\widetilde{\mathcal{B}}$ is an equilibrium partition. Furthermore, it must be that $\mu_F(\widetilde{B}_k) > \mu_F(B_k)$: this is because by construction, $\mathbb{E}\left[\omega \, | \, \omega \in B_k \right] = \mathbb{E}\left[\omega \, | \, \omega \in \widetilde{B}_k \right] = \gamma_k$ and $\int_{\widetilde{B}_k} \omega \ddd F(\omega) > \int_{B_k} \omega \ddd F(\omega)$. Consequently, the sender's payoff is strictly higher in this new equilibrium since $j < k$, which contradicts the assumption that $\mathcal{B}$ is a partition associated with a sender-preferred equilibrium. 
\end{proof}

\begin{claim} \label{claim:below-above}
Let $\mathcal{B}$ be a barely obedient laminar partition, and let $i, j$ be two unskipped actions with $j < i$. Then if $\mu_F(\conv{(B_i)} \cap \conv{(B_j)}) = 0$, then every $\omega' \in B_j$ and $\omega'' \in B_{i}$ satisfy $\omega' \le \omega''$. 
\end{claim}

\begin{proof}[Proof of \cref{claim:below-above}]

    Because $\mu_F(\conv{(B_i)} \cap \conv{(B_j)}) = 0$, either $\omega' \le \omega''$ for every $\omega' \in B_j$ and $\omega'' \in B_{i}$, or $\omega' \ge \omega''$ for every $\omega' \in B_j$ and $\omega'' \in B_{i}$. The second possibility, however, cannot be true: by obedience, $j < i$ implies that $\underline{b}_j \le \gamma_{i}$, and $\mathbb{E}[\omega \sep \omega \in B_{i}] \ge \gamma_{i}$. Consequently, $\overline{b}_{i} > \gamma_{i} \ge \underline{b}_j$, a contradiction. Therefore, it must be that $\omega' \le \omega''$ for every $\omega' \in B_j$ and $\omega'' \in B_{i}$. 
\end{proof}

We are now ready to prove \autoref{thm:sender-pref-eq-characterization}.

\begin{proof}[Proof of \autoref{thm:sender-pref-eq-characterization}]
We first show that $\mathcal{B}$ must be barely obedient. Suppose to the contrary that $\mathcal{B}$, a laminar partition associated with a sender-preferred equilibrium, has $\mathbb{E}[\omega \mid \omega \in B_k] > \gamma_k$ for some non-null $k > j$. By \autoref{lemma:eqm-partition}, $\mathcal{B}$ is both obedient and revelation-proof. Let $\underline{b}_j \coloneqq \min B_j$; because $f > 0$, there exists $x \in B_j$ with $x > \underline{b}_j$ such that $\mathbb{E}[\omega \mid \omega \in [\underline{b}_j, x] \cup B_k] \ge \gamma_k$, and $\mathbb{E}[\omega \mid \omega \in B_j \smallsetminus [\underline{b}_j, x]] > \gamma_j$. Now consider the partition $\widetilde{\mathcal{B}}$ where $\widetilde{B}_j = B_j \smallsetminus [\underline{b}_j, x]$, $\widetilde{B}_k = [\underline{b}_j, x] \cup B_k$, and $\widetilde{B}_i = B_i$ for $i \ne j,k$. $\widetilde{\mathcal{B}}$ is obedient and revelation-proof because $\mathcal{B}$ is, and the sender's ex-ante payoff from $\widetilde{\mathcal{B}}$, $V(\widetilde{\mathcal{B}})$, satisfies
\[V(\widetilde{\mathcal{B}}) - \overline{V} = \left[F(x) - F(\underline{b}_j)\right] (u_k - u_j) > 0,\]
a contradiction.

Next, we show that $B_i$ is the union of at most $\max\{i-1,1\}$ closed intervals for all $i \in N$.
Recall that a laminar partition is defined by \eqref{eqn:laminar-definition}. Because $\conv{(B_k)}$ must be an interval for all $k$, $\cup_{k < i} \conv{(B_k)}$ is the union of at most $i-1$ intervals. By \autoref{claim:lowest-element}, since 
$\mathcal{B}$ is associated with a sender-preferred equilibrium, $\inter{(B_1)} \cap B_k = \varnothing$ for all $k \ge 2$. Thus, for any $i \ge 2$, by taking out $\cup_{k < i} \conv{(B_k)}$ from $\conv{(B_i)}$, at most $i-2$ intervals are removed, and hence the remainder, namely $\conv{(B_i)} \smallsetminus \cup_{k < i} \conv{(B_k)}$, must be the union of at most $\max\{i-1,1\}$ intervals. By taking closure, $B_i$ is also the union of at most $\max\{i-1,1\}$ intervals. 

Finally, we prove the equivalence between revelation proofness of a barely obedient laminar partition $\mathcal{B}$ and the condition that $\max B_i \leq \gamma_{i+1}$ for all unskipped $i$ such that $i+1$ either nests $i$ or is skipped. For convenience, for any $B_i \in \mathcal{B}$, denote $\underline{b}_i \coloneqq \min B_i$ and $\overline{b}_i \coloneqq \max B_i$.\footnote{If $B_i = \varnothing$, we let $\underline{b}_i = \overline{b}_i = 0$.}

Suppose first that $\max B_i \leq \gamma_{i+1}$ for all unskipped $i$ such that $i+1$ either nests $i$ or is skipped, and we show that for any action $i \in N$, $B_i \subseteq [0,\gamma_{i+1}]$. Since $\mathcal{B}$ is laminar, it is without loss of generality to assume that if action $i$ is skipped, $B_i \subseteq [0,\gamma_{i+1}]$. 
    
    Now suppose $i$ is not skipped. There are two possibilities: either action $i+1$ is skipped or not. If $i+1$ is skipped, by assumption $\overline{b}_i \le \gamma_{i+1}$, implying $B_i \subseteq [0,\gamma_{i+1}]$. If instead $i+1$ is not skipped, there are two cases: either $B_i \subseteq \conv{(B_{i+1})}$ or not. If $B_i \subseteq \conv{(B_{i+1})}$, then by assumption $\overline{b}_{i} \le \gamma_{i+1}$, implying that $B_i \subseteq [0,\gamma_{i+1}]$. 

    Suppose instead that $B_i$ is not a subset of $\conv{(B_{i+1})}$, and suppose to the contrary that $\overline{b}_{i} > \gamma_{i+1}$. Because $\mathcal{B}$ is laminar, since $B_i$ is not a subset of $\conv{(B_{i+1})}$ (which means that $\conv{(B_{i+1})}$ does not nest $\conv{(B_{i})}$), it must be that $\mu_F(\conv{(B_{i})} \cap \conv{(B_{i+1})}) = 0$. By \cref{claim:below-above}, for every $\omega' \in B_{i}$ and every $\omega'' \in B_{i+1}$, $\omega' \le \omega''$. This implies that $\omega \ge  \overline{b}_{i} > \gamma_{i+1}$ for every $\omega \in B_{i+1}$, violating bare obedience, a contradiction. Thus, $\overline{b}_{i} \le \gamma_{i+1}$, which implies that $B_i \subseteq [0,\gamma_{i+1}]$. 

    For the other direction, we prove the contrapositive. If the condition is violated, that is, if $\max B_i > \gamma_{i+1}$ for some unskipped $i$ such that $i+1$ either nests $i$ or is skipped, then $B_i$ is not a subset of $[0,\gamma_{i+1}]$. Thus, $\mathcal{B}$ must violate revelation proofness.
\end{proof}

\subsection{Proofs and Details for \autoref{s:imp}}
We first introduce the notion of a bi-pooling distribution.

\begin{definition}[Bi-pooling distribution]
	A distribution $G \in \text{MPC}(F)$ is a \ul{bi-pooling distribution} if there exists a collection of pairwise disjoint intervals $\left\{\left(\underline{\omega}_{i}, \overline{\omega}_{i}\right)\right\}_{i \in \mathcal{I}}$ such that
	\begin{itemize}[noitemsep,topsep=1.5pt]
	    \item for all $i \in \mathcal{I}$, $G\left(\overline{\omega}_{i}\right) - G\left(\underline{\omega}_{i}\right) = F\left(\overline{\omega}_{i}\right) - F\left(\underline{\omega}_{i}\right)$ and $\left|\supp \left(G |_{\left(\underline{\omega}_{i}, \overline{\omega}_{i}\right)}\right)\right| \le 2$;\footnote{For any cumulative distribution function $H$, let $H|_{[c,d]}$ denote the restriction of $H$ to $[c,d] \subseteq [0,1]$.}
	    \item $\left.G\right|_{[0,1] \smallsetminus \cup_{i \in \mathcal{I}}\left(\underline{\omega}_{i}, \overline{\omega}_{i}\right)}=\left.F\right|_{[0,1] \smallsetminus \cup_{i \in \mathcal{I}}\left(\underline{\omega}_{i}, \overline{\omega}_{i}\right)}$.
	\end{itemize}
 In particular, $\left(\underline{\omega}_{i}, \overline{\omega}_{i}\right)$ is called a \ul{pooling interval} if $\left|\supp \left(G |_{\left(\underline{\omega}_{i}, \overline{\omega}_{i}\right)}\right)\right| = 1$, and it is called a \ul{bi-pooling interval} if $\left|\supp \left(G |_{\left(\underline{\omega}_{i}, \overline{\omega}_{i}\right)}\right)\right| = 2$. 

We call a bi-pooling distribution $G_B$ that solves the commitment problem \eqref{id_problem} a \ul{bi-pooling solution}.
\end{definition}

The following observation is useful. 

\begin{lemma}[\citeauthor{candogan2019optimality}, \citeyear{candogan2019optimality}] \label{lemma:bipool_chara}
    Every bi-pooling solution to the commitment problem satisfies $\supp{(G_B)} \subseteq \left[0,\widetilde{\gamma}\right] \cup \left\{\gamma_{i}\right\}_{i=2}^{n}$, where $\widetilde{\gamma} \in [0,1]$ is such that $\widetilde{\gamma} \le \gamma$ for all $\gamma \in \supp{(G_B)} \cap \left\{\gamma_{i}\right\}_{i=2}^{n}$.
\end{lemma}

An important consequence of \cref{lemma:bipool_chara} is that every signal realization can be identified by the action it induces. In particular, every $x \in \supp{(G_B)} \cap [0, \widetilde{\gamma}]$ induces the lowest unskipped action, and $\gamma_i$ induces action $i$ for each $i=2,\ldots,n$. Moreover, if a bi-pooling solution is associated with a bi-pooling partition, then \cref{lemma:bipool_chara} implies that the partition must be barely obedient. 

\subsubsection{Proof of \autoref{thm:id-sol}}
    By results in  \cite{kms} and \cite{absy}, the commitment problem \eqref{id_problem} admits a bi-pooling solution. 

    By \cref{lemma:bipool_chara}, in any bi-pooling solution $G_B$, if $\supp(G_B) \smallsetminus \left\{\gamma_{i}\right\}_{i=2}^{n}$ is nonempty but not a singleton, there must exist an interval $[0, d]$ on which action $1$ is recommended, and $[d,1]$ is comprised of pooling intervals and/or bi-pooling intervals; in this case, $[0,d]$ can be viewed as a pooling interval. Otherwise, $[0,1]$ comprises pooling and/or bi-pooling intervals. The fact that every bi-pooling solution $G_B$ is associated with a bi-pooling partition therefore follows from Lemma 4 in \citet{absy}: on every pooling interval, a single signal realizes, which induces a single action and hence ties to a single partitional element. On every bi-pooling interval $(\underline{\omega}, \overline{\omega})$, there exists $(\underline{q}, \overline{q}) \subseteq (\underline{\omega}, \overline{\omega})$ such that one (deterministic) signal realizes when $\omega \in (\underline{q}, \overline{q})$, and another signal realizes when $\omega \in (\underline{\omega}, \underline{q}) \cup  (\overline{q}, \overline{\omega})$, with the latter signal recommending a higher action. Then $[\underline{q}, \overline{q}]$ and $[\underline{\omega}, \underline{q}] \cup  [\overline{q}, \overline{\omega}]$ correspond to two partitional elements $B_\ell$ and $B_h$ with $\ell < h$, respectively. This observation and \cref{lemma:bipool_chara} together imply that the bi-pooling partition is barely obedient. 

    Fix a bi-pooling solution $G_B$. To see that there is a unique barely obedient bi-pooling partition associated with it, we first note that every pooling interval necessarily ties to a unique partitional element. Hence, if a bi-pooling solution admits two distinct bi-pooling partitions, $\mathcal{B}^1$ and $\mathcal{B}^2$, they must differ on a bi-pooling interval $(\underline{\omega},\overline{\omega})$. Then there must exist $\{B_\ell, B_h\}$ and $\{B'_\ell, B'_h\}$ such that $B_\ell = [\underline{b}_\ell, \overline{b}_\ell]$, $B_h = [\underline{\omega}, \underline{b}_\ell] \cup[\overline{b}_\ell, \overline{\omega}]$, $B'_\ell = [\underline{b}'_\ell, \overline{b}'_\ell]$, and $B'_h = [\underline{\omega}, \underline{b}'_\ell] \cup[\overline{b}'_\ell, \overline{\omega}]$, with either $\underline{b}'_\ell \ne \underline{b}_\ell$, or $\overline{b}'_\ell \ne \overline{b}_\ell$, or both. Since $f>0$ and $\underline{\omega}$ and $\overline{\omega}$ are fixed, either $\mathbb{E}[\omega \sep \omega \in B_\ell] \ne \gamma_\ell$, or $\mathbb{E}[\omega \sep \omega \in B'_\ell] \ne \gamma_\ell$, or both, a contradiction. 

    Finally, we show that the commitment solution is unique in generic communication environments. 
    For any bi-pooling solution $G_B$, let $\mu_{i} := G_B(\gamma_{i+1}) - G_B(\gamma_{i})$ for each $i=2, \ldots, n$. If $\supp{(G_B)} \smallsetminus \left\{\gamma_{i}\right\}_{i=1}^{n} = \varnothing$, set $\mu_1 = G_B(\gamma_{2})$ and $\widetilde{\mu} = 0$; otherwise, let $\widetilde{\mu} = G_B(\widetilde{\gamma})$ and $\mu_1 = G_B(\gamma_{2}) - \widetilde{\mu}$; \cref{lemma:bipool_chara} then indicates that every bi-pooling solution is identified by $\{\widetilde{\mu}\} \cup \{\mu_i\}_{i=1}^{n}$. By Lemma D.1 in \cite{candogan2022persuasion}, the condition that no three elements of the collection of points $\{(\gamma_i, u_i)\}_{i=1}^{n}$ are collinear implies that any two bi-pooling solutions of problem \eqref{id_problem} induce the same collection of $\{\widetilde{\mu}\} \cup \{\mu_i\}_{i=2}^{n}$. Then, because every bi-pooling solution is a mean-preserving spread of $F$, any two such bi-pooling solutions must be identical. As noted in \cite{kms} and \cite{absy}, every extreme point of the set of solutions to problem \eqref{id_problem} is a bi-pooling solution, and hence the solution to problem \eqref{id_problem} must be a unique bi-pooling solution.

\subsubsection{Proof of \autoref{prop:bipooling-implementable-iff}}
Since the bi-pooling partition is barely obedient, the ``if'' direction follows from \cref{lemma:eqm-partition}. For the other direction, suppose first that a bi-pooling solution $G_B$ is implementable. Since the commitment payoff is an upper bound on equilibrium payoffs, $G_B$ is the PMD induced by a sender-preferred equilibrium $(\sigma, \tau, \beliefm)$. By \autoref{thm:id-sol}, $G_B$ is associated with a unique barely obedient bi-pooling partition $\mathcal{B}$; an argument analogous to the one in the proof of \autoref{lemma:sender-preferred-laminar} shows that $\mathcal{B}$ is revelation-proof.

\subsubsection{Proof of \autoref{prop:sufficient}}
We use the following result by \cite{absy} to prove \cref{prop:sufficient}. Say that $\{\gamma_\ell,\gamma_h\} $ is a \ul{feasible} bi-pooling support for the interval $(\underline{\omega},\overline{\omega})$ (or just \ul{feasible} for simplicity) if there exists a mean preserving contraction of $F|_{[\underline{\omega},\overline{\omega}]}$ whose support is $\{\gamma_\ell,\gamma_h\}$.

\begin{lemma}[\citeauthor{absy}, \citeyear{absy}] \label{bpf}
	Fix an interval $(\underline{\omega},\overline{\omega})$, and let $y(\gamma_\ell)$ satisfy $\mathbb{E}[\omega \mid \omega \in [\underline{\omega}, y(\gamma_\ell))] = \gamma_\ell$. Then $\{\gamma_\ell,\gamma_h\}$ is feasible for the interval $(\underline{\omega},\overline{\omega})$ if and only if
	\begin{itemize}[itemsep=0.1pt, topsep=2.5pt]
		\item[(i)] $\underline{\omega} \le \gamma_\ell \le \mathbb{E}[\omega \mid \omega \in [\underline{\omega}, \overline{\omega}]] \le \gamma_h \le \overline{\omega}$, and
		\item[(ii)] $\mathbb{E}[\omega \mid \omega \in [y(\gamma_\ell), \overline{\omega}]] \ge \gamma_h$.
	\end{itemize}
\end{lemma}

\begin{proof}[Proof of \cref{prop:sufficient}]
    Suppose to the contrary that there exists a bi-pooling solution $G_B$ that is not implementable. Then by \cref{cor:bip_imp_2}, the bi-pooling partition of $G_B$, $\mathcal{B}$, must violate the condition therein. First suppose that there exists an unskipped action $i$ such that $\overline{b}_i > \gamma_{i+1}$. There are two cases.

\paragraph{Case 1.} $B_i$ is an interval; i.e., $B_i = [\underline{b}_i, \overline{b}_i]$. There are two subcases:
\begin{itemize}[noitemsep, topsep=1pt]
	\item[(I)] $\gamma_{i-1} \in  \inter(B_i)$. 
	
	For $z \in (\gamma_{i+1}, \overline{b}_i)$, let $h(z)$ solve $\mathbb{E}[\omega \mid \omega \in [\underline{b}_i, h(z)] \cup [z, \overline{b}_i]] = \gamma_{i+1}$.
By \cref{bpf}, for $z$ close enough to $\overline{b}_i$, $\{\gamma_{i-1}, \gamma_{i}\}$ is feasible for $[h(z), z]$. 
Consequently, the sender's payoff on $B_i$, as a function of $z$, is
\begin{align*}
	P(z) \coloneqq~ & u_{i+1}[F(\overline{b}_{i})-F(z)+F(h(z))-F(\underline{b}_{i})] + \\
    & \left[\frac{m(z)-\gamma_{i-1}}{\gamma_{i}-\gamma_{i-1}} u_{i}+\frac{\gamma_{i}-m(z)}{\gamma_{i}-\gamma_{i-1}} u_{i-1}\right](F(z)-F(h(z)))
\end{align*}
where $m(z):=\mathbb{E}[\omega \mid \omega \in[z, h(z)]]$. To show that this is a profitable deviation, it suffices to show that $P'(\overline{b}_i) := \lim_{z \nearrow \overline{b}_i} P'(z) < 0$. To this end, we first calculate 
\begin{align*}
    P'(z) 
    & = \frac{f(z)}{\gamma_{i}-\gamma_{i-1}} \frac{z - h(z)}{\gamma_{i+1} - h(z)} \left[(u_{i} - u_{i-1}) \gamma_{i+1} + (u_{i+1} - u_{i}) \gamma_{i-1} - (u_{i+1} - u_{i-1}) \gamma_i\right]. 
\end{align*}
Letting $z \nearrow \overline{b}_i$, then $h(z) \searrow \underline{b}_i$, and
\begin{align*}
	P'(\overline{b}_i) & = \frac{f(\overline{b}_i)}{\gamma_{i}-\gamma_{i-1}} \frac{\overline{b}_i - \underline{b}_i}{\gamma_{i+1} - \underline{b}_i} \left[ (u_{i} - u_{i-1}) \gamma_{i+1} + (u_{i+1} - u_{i}) \gamma_{i-1} - (u_{i+1} - u_{i-1}) \gamma_i\right];
\end{align*}
$P'(\overline{b}_j) < 0$ if and only if
$
(u_{i+1} - u_{i-1}) \gamma_i - (u_{i} - u_{i-1}) \gamma_{i+1} - (u_{i+1} - u_{i}) \gamma_{i-1} > 0,
$
and this is equivalent to 
\[
\frac{u_{i+1} - u_{i}}{\gamma_{i+1} - \gamma_i} > \frac{u_{i} - u_{i-1}}{\gamma_{i} - \gamma_{i-1}},
\]
which is implied by \eqref{eq:convexity}.
	
	\item[(II)] $\gamma_{i-1} \notin  \inter(B_i)$ (and therefore $\gamma_{i - 1} \le \underline{b}_i$). 
    
    Define $B_i^\varepsilon = [\underline{b}_i+\varepsilon, \overline{b}_i]$; for small enough $\varepsilon > 0$, $\{\gamma_i, \gamma_{i+1}\}$ is feasible for $B^\varepsilon_i$. Let $m(\varepsilon)$ denote the mean of $B^\varepsilon_i$; note that $m(0) = \gamma_i$. Then the sender's payoff as a function of $\varepsilon$ on $B_i$ is 
	\begin{align*}
		P(\varepsilon) & \coloneqq \left[\frac{m(\varepsilon)-\gamma_{i}}{\gamma_{i+1}-\gamma_{i}} u_{i+1}+\frac{\gamma_{i+1}-m(\varepsilon)}{\gamma_{i+1}-\gamma_{i}} u_{i}\right][F(\overline{b}_i)-F(\underline{b}_i+\varepsilon)] + u_{i-1} [F(\underline{b}_i+\varepsilon) - F(\underline{b}_i)].
	\end{align*}
	To show that this is a profitable deviation, it suffices to show that $P'(0) := \lim_{\varepsilon \searrow 0} P'(\varepsilon) > 0$. Algebra reveals that
	\begin{align}
		P'(\varepsilon) & = f(\underline{b}_i+\varepsilon) \left[u_{i-1} - \frac{\gamma_{i+1} u_{i} - \gamma_i u_{i+1}}{\gamma_{i+1}-\gamma_{i}} -(\underline{b}_i-\varepsilon)\frac{u_{i+1}-u_{i}}{\gamma_{i+1}-\gamma_{i}}\right]; \nonumber
	\end{align}
	consequently, as $\varepsilon \searrow 0$,
	\begin{align}
		P'(0)
		& = \frac{f(\underline{b}_i)}{\gamma_{i+1}-\gamma_{i}} \left[(\gamma_{i}-\underline{b}_i)(u_{i+1} - u_{i}) - (\gamma_{i+1}-\gamma_{i})(u_{i} - u_{i-1})\right]. \label{dgb}
	\end{align}
	Because $f(\underline{b}_i) > 0$ and $\gamma_{i+1}-\gamma_{i} > 0$ by assumption, the sign of $P'(0)$ is the same as the sign of term in the square brackets in the right-hand side of \eqref{dgb}. Thus, $P'(0) > 0$ if and only if 
	\[
	\frac{u_{i+1} - u_{i}}{\gamma_{i+1} - \gamma_i} > \frac{u_{i} - u_{i-1}}{\gamma_i - \underline{b}_i}.
	\] 
	Because $\gamma_{i+1} < \overline{b}_i$, $h(\gamma_i;\gamma_{i+1}) > \underline{b}_i$. Thus, \eqref{eq:convexity} implies the inequality above.
\end{itemize} 

\paragraph{Case 2.} There is a partitional element $B_j$ with $j<i$ such that $B_j \subseteq \conv{(B_i)}$. 
In this case, $B_j = [\underline{b}_j, \overline{b}_j]$, and $B_i = [\underline{b}_i, \underline{b}_j] \cup [\overline{b}_j, \overline{b}_i]$. WLOG, assume that $\underline{b}_i < \underline{b}_j$; by \cref{bpf}, $\{\gamma_j, \gamma_i\}$ is feasible for $[\underline{b}_i, \overline{b}_i]$, and
\begin{equation} \label{ste}
    \mathbb{E}[\omega \mid \omega \in [y(\gamma_j), \overline{b}_i]] > \gamma_i.
\end{equation}
For $z \in (\max\{\gamma_{i+1}, \overline{b}_j\}, \overline{b}_i)$, let $h(z)$ solve $\mathbb{E}\big[\omega \mid \omega \in [\underline{b}_i, h(z)] \cup [z, \overline{b}_i]\big] = \gamma_{i+1}$.
By \eqref{ste} and \cref{bpf}, for $z$ close enough to $\overline{b}_i$, $\{\gamma_j, \gamma_i\}$ is feasible for $[h(z), z]$. Consequently, one can find a profitable deviation similar to in Case 1 (I) \emph{mutatis mutandis}.

\bigskip

Now suppose instead that there exists a pair of partitional elements $B_i$ and $B_{i+1}$ with $B_i \subseteq \conv{(B_{i+1})}$, and $\gamma_{i+1} < \overline{b}_i$. Since $B_i = [\underline{b}_i, \overline{b}_i]$, this case is isomorphic to Case 1 when action $i+1$ is skipped. Hence, we can similarly find a profitable deviation.

Therefore, every bi-pooling solution must satisfy the condition in \cref{cor:bip_imp_2}, and hence implementable. This completes the proof. 
\end{proof}

\subsubsection{Proof of \autoref{cor:inc-dens}}
By definition of $h(\gamma_i; \gamma_{i+1})$, when $f$ is increasing, $\gamma_{i+1} - \gamma_{i} \le \gamma_{i}-h\left(\gamma_{i} ; \gamma_{i+1}\right)$. Then since $\gamma_{i+1} - \gamma_{i} \le \gamma_{i} - \gamma_{i-1}$, it must be that $\gamma_{i+1}-\gamma_{i} \le \min \left\{\gamma_{i}-\gamma_{i-1}, \gamma_{i}-h\left(\gamma_{i} ; \gamma_{i-1}\right)\right\}$. 
Now there are two cases. If $u_{i+1} - u_{i} > u_{i} - u_{i-1}$, \eqref{eq:convexity} must hold; then by \cref{prop:sufficient}, every bi-pooling solution can be implemented. If instead $\gamma_{i+1} - \gamma_{i} < \gamma_{i} - \gamma_{i-1}$, it must be that $\gamma_{i+1}-\gamma_{i} < \min \left\{\gamma_{i}-\gamma_{i-1}, \gamma_{i}-h\left(\gamma_{i} ; \gamma_{i-1}\right)\right\}$; this inequality and $u_{i+1} - u_{i} \ge u_{i} - u_{i-1}$ together imply \eqref{eq:convexity}. Again by \cref{prop:sufficient},  every bi-pooling solution can be implemented.

\section{Supplementary Appendix}
This appendix is organized as follows:
\begin{itemize}[noitemsep]
    \item \autoref{appendix:supplementary-ternary} proves supplementary results for the special case in which the receiver has three actions.
    \item \autoref{appendix:remaining-proofs} contains the proofs omitted from Appendix \ref{section:appendix}.
    \item \autoref{refinement} establishes the robustness of the set of partitional equilibria to two standard refinements: the never-a-weak-best-response (NWBR) criterion \citep{chokreps1987} and the Grossman--Perry--Farrell refinement \citep{bertomeu2018verifiable}.
\end{itemize}

\subsection{Supplementary Results for Ternary Actions} \label{appendix:supplementary-ternary}
In this section, we study the special case where the receiver has three actions: $N = \{1, 2, 3\}$. Recall that $\{\gamma_L,\gamma_H\}$ is feasible for the interval $[\underline{\omega},\overline{\omega}]$ if there exists a mean preserving contraction of $F|_{[\underline{\omega},\overline{\omega}]}$ whose support is $\{\gamma_L,\gamma_H\}$. 

\begin{claim} \label{imc}
    If there does not exist $y \in [0,\gamma_2]$ such that $\{\gamma_2, \gamma_3\}$ is feasible for $[y,1]$, then every bi-pooling solution is implementable.
\end{claim}

\begin{proof}[Proof of \autoref{imc}]
    By Lemma 6, $\{\gamma_2,\gamma_3\}$ is feasible for the interval $[y,1]$ if and only if
	\begin{itemize}
		\item[(i)] $y \le \gamma_2 \le m(y) \le \gamma_3 \le 1$, and
		\item[(ii)] $\mathbb{E}[\omega \, | \, \omega \in [\eta(\gamma_2; y), 1]] \ge \gamma_3$,
	\end{itemize}
	where $m(y) \coloneqq \mathbb{E}[\omega \, | \, \omega \in [y, 1]]$, and $\eta(\gamma_2; y)$ is such that $\mathbb{E}[\omega \, | \, \omega \in [y, \eta(\gamma_2)] = \gamma_2$. Then if there does not exist $y \in [0,\gamma_2]$ such that $\{\gamma_2, \gamma_3\}$ is feasible for $[y,1]$, there are two cases:
	\begin{itemize}
	    \item[(a)] (i) fails to hold for all $y \in [0,\gamma_2]$;
	    \item[(b)] (i) holds for some $y \in [0,\gamma_2]$, but (ii) fails for all such $y$'s.
	\end{itemize}
	For Case (a), the only possibility is that $m(y) \ge \gamma_3$. If this is the case, every bi-pooling solution $G$ has $\supp{(G)} \subseteq [\gamma_3, 1]$, and hence the unique bi-pooling partition has $B_3 = [0,1]$, and revelation proofness holds. 
	
	For Case (b), let us introduce some notation first. For each $i=1,2$, if there exists $h \ge 0$ such that $\mathbb{E}[\omega \, | \, \omega \in [h,1]] = \gamma_i$, set $h(\gamma_i) = h$; otherwise let $h(\gamma_i) = 0$. Then for every $y \in [0, \gamma_2] \cap [h(\gamma_2), h(\gamma_3)]$, (i) holds. If (ii) fails for all such $y$'s, it must be that $\mathbb{E}[\omega \, | \, \omega \in [\eta(\gamma_2; h(\gamma_2)), 1]] < \gamma_3$. Note that this is not possible if $h(\gamma_2) > 0$: if this is the case, $\eta(\gamma_2; h(\gamma_2)) = 1$ by definition, so it must be that $\mathbb{E}[\omega \, | \, \omega \in [\eta(\gamma_2; h(\gamma_2)), 1]] \ge \gamma_3$, a contradiction. Consequently, in Case (b), (ii) must fail for $y = 0$, and hence every bi-pooling solution $G$ has $G(\gamma_3) - G(\gamma_3^-) = 1 - F(h(\gamma_3))$, and $\supp{(G)} \subseteq [\gamma_2,\gamma_3]$. Thus, the unique bi-pooling partition associated with $G$ has $B_2 = [0, h(\gamma_3)]$ and $B_3 = [h(\gamma_3), 1]$. Since $h(\gamma_3) < \gamma_3$ by definition, revelation proofness must hold. This completes the proof.
\end{proof}

\begin{claim} \label{euc}
    All bi-pooling solutions to the information design problem are associated with the same bi-pooling partition. 
\end{claim}

\begin{proof}[Proof of \autoref{euc}]
    It can be readily seen from the proof of \autoref{imc} that, if there does not exist $y \in [0,\gamma_2]$ such that $\{\gamma_2, \gamma_3\}$ is feasible for $[y,1]$, then all bi-pooling solutions are associated with the same bi-pooling partition. Now suppose there exists $y \in [0,\gamma_2]$ such that $\{\gamma_2, \gamma_3\}$ is feasible for $[y,1]$. Let $Y$ denote the set of such $y$'s; Lemma 6 implies that $Y$ is a closed subset of $[0,\gamma_2]$. Consequently, the commitment payoff can be identified by the lower endpoint of the bi-pooling interval. Hence, each of them corresponds to a point in $Y$ that maximizes the sender's ex-ante payoff (recall that $u_1 = 0$):
    \[
    \pi(y) = (1-F(y))\left[\frac{\gamma_3-m(y)}{\gamma_3-\gamma_2} u_2+\frac{m(y)-\gamma_2}{\gamma_3-\gamma_2} u_3\right].
    \]
    Taking derivative, 
    \[
    \pi^{\prime}(y)=-f(y)\left[\frac{\gamma_3-m(y)}{\gamma_3-r_{1}} u_2+\frac{m(y)-\gamma_2}{\gamma_3-\gamma_2} u_3\right]+(1-F(y)) \frac{u_3-u_2}{\gamma_3-u_2} m^{\prime}(y),
    \]
    where
    \[
    m^{\prime}(y)=\frac{(m(y)-y) f(y)}{1-F(y)}.
    \]
    Consequently,
    \begin{align*}
        \pi'(y) & = \frac{f(y)}{\gamma_3 - \gamma_2} \left[(m(y)-y)\left(u_3-u_2\right)-\left(\gamma_3-m(y)\right) u_2-\left(m(y)-\gamma_2\right) u_3\right], \\
        & = \frac{f(y)}{\gamma_3 - \gamma_2} \left[u_3\left(\gamma_2-y\right)-u_2\left(\gamma_3-y\right)\right]
    \end{align*}
    and its sign is determined by the terms between the squared brackets, which implies that $\pi$ is single-peaked in $y$. Therefore, there must exist a unique $z$ that maximizes $\pi(y)$ on $Y$. As a consequence, all bi-pooling solutions are associated with the same bi-pooling partition $\mathcal{B}$ with $B_1 = [0,z]$, $B_2 = [\underline{b}_2, \bar{b}_2]$, and $B_3 = [z, \underline{b}_2] \cup [\bar{b}_2,1]$. 
\end{proof}

\begin{claim} \label{ntm}
    If no commitment solution is implementable, the sender-preferred laminar PE $\mathcal{B}$ is such that $B_1 = [0, y]$, $B_2= [h,b]$, and $B_3 = [y,h] \cup [b,1]$, where $\gamma_2 < b \le \gamma_3$, and $h > 0$ and $y \ge 0$ are defined by 
    \begin{align} \label{ime}
        \mathbb{E}\left[\omega \mid \omega \in\left[h, b\right]\right]=\gamma_2 \quad 
	\text{and} \quad \mathbb{E}\left[\omega \mid \omega \in\left[y, h\right] \cup\left[b, 1\right]\right]=\gamma_3.
    \end{align}
Furthermore, 
\begin{equation}\label{rie}
    u_3 \le \frac{\gamma_3 - y}{\gamma_2 - y} u_2.
\end{equation}
\end{claim}

\begin{proof}[Proof of \autoref{ntm}]
    By Theorem 2, if $B_2$ has nonempty interior, it must be that $\mathbb{E}[\omega \mid \omega \in B_2]  = \gamma_2$, and $\mathbb{E}[\omega \mid \omega \in B_3]  = \gamma_3$.\footnote{$B_3$ must have nonempty interior, or else revelation proofness cannot be satisfied. Also by revelation proofness, $b \le \gamma_3$.} Then to obtain the statement it suffices to show three things: (1) $b > \gamma_2$, namely $B_2$ has nonempty interior; (2) there exists such an $h$; and (3) there exists such a $y$. \par 
    
    We show that $B_2$ has nonempty interior first. Suppose to the contrary that $\inter(B_2) = \varnothing$, then by Claim 7, there are two cases: $B_3 = [0,1]$, and $B_3 = [z, 1]$ for some $z > 0$. If $B_3 = [0,1]$, since the sender attains highest possible payoff in equilibrium, it must be that a commitment outcome is implementable, a contradiction. If instead $B_3 = [z, 1]$ for some $z > 0$, it must be that $z \le \gamma_2$, as otherwise recommending action $2$ on $[z, \gamma_2]$ is a profitable deviation. Consequently, this bi-pooling partition is revelation-proof, and hence a commitment outcome is implementable, again a contradiction. Therefore, it must be that $B_2$ has nonempty interior.
    
    To see that there exists such an $h$, we first claim that $\mathbb{E}[\omega \mid \omega \in [0,\gamma_3]] \le \gamma_2$. Suppose not, so 
    \begin{equation} \label{jbm}
        \mathbb{E}[\omega \mid \omega \in [0,\gamma_3]] > \gamma_2.
    \end{equation}
    Without loss of generality, assume that there exists a bi-pooling solution features $[y,1]$ bi-pooled to $\{\gamma_2, \gamma_3\}$ for some $y \in [0,\gamma_2]$.\footnote{If such $y$ does not exist, by \autoref{imc}, the bi-pooling solution must be implementable, a contradiction.} Consequently, there exist $\bar{b}_1$ and $\underline{b}_1$ with $y \le \underline{b}_1 \le \bar{b}_1$ such that the (unique) bi-pooling partition associated with the bi-pooling solution $\mathcal{B}$ is given by $B_1 = [0,y]$, $B_2 = [\underline{b}_1,\bar{b}_1]$, and $B_3 = [y,\underline{b}_1] \cup [\bar{b}_1,1]$. Then because $\mathbb{E}[\omega \mid \omega \in B_2] = \gamma_2$ and $\underline{b}_1 \ge 0$, we must have $\bar{b}_1 \le \gamma_3$ by (\ref{jbm}). Thus, the bi-pooling solution must be implementable, a contradiction. As a consequence, it must be that $\mathbb{E}\left[\omega \mid \omega \in\left[h, b\right]\right]=\gamma_2$; it remains to show that $h > 0$. If instead $h = 0$, then $\mathbb{E}[\omega \mid \omega \in [0,\gamma_3]] = \gamma_2$. Consequently, it must be that $B_2 = [0,\gamma_3]$ and $B_3 = [\gamma_3,1]$. This cannot be optimal: for any $\varepsilon \in (0, \gamma_2)$, define $\hat{B}_1 = [\varepsilon, \gamma_3]$, and $\hat{B}_2 = [0, \varepsilon] \cup [\gamma_3, 1]$. Then for $\varepsilon$ small enough, $\mathbb{E}[\omega \mid \omega \in \hat{B}_i] \ge \gamma_i$ for each $i = 1,2$, and the sender's ex-ante payoff is strictly higher. Thus, it must be that $h > 0$.
    
    To show that there exists such a $y$, it suffices to show that $\mathbb{E}[\omega \mid \omega \in [0,h] \cup [b,1]] \le \gamma_3$. Suppose not, so $\mathbb{E}[\omega \mid \omega \in [0,h] \cup [b,1]] > \gamma_3$. Let $\delta > 0$ be small enough, and let $\epsilon(\delta)$ be such that
    \[
    \mathbb{E}\left[\omega \mid \omega \in \left[h+\varepsilon(\delta), b-\delta\right]\right] = \gamma_2.
    \]
    Now define $\Tilde{B}_1 = \left[h+\varepsilon(\delta), b-\delta\right]$, and $\Tilde{B}_2 = [0,h+\varepsilon(\delta)] \cup [b-\delta,1]]$. Because the density $f$ is strictly positive, for small enough $\delta$, $\mathbb{E}[\omega \mid \omega \in \Tilde{B}_2] \ge \gamma_3$, and $\mu_F(\Tilde{B}_2) > \mu_F(B_3)$. This creates a profitable deviation to the sender without violating revelation proofness. 
    
    Finally, to show that (\ref{rie}) must hold, suppose to the contrary that 
    \[u_3 > \frac{\gamma_3 - y}{\gamma_2 - y} u_2.\]
    An argument analogous to Case 1 (II) in the proof of Proposition 2 shows that the sender has a profitable deviation, and hence the bi-pooling partition $\mathcal{B}$ with $B_1 = [0, y]$, $B_2= [h,b]$, and $B_3 = [y,h] \cup [b,1]$ cannot be associated with a sender-preferred laminar PE, a contradiction.
\end{proof}

By \autoref{ntm}, the sender's ex-ante payoff in a sender-preferred equilibrium can be written as
\[
\bar{V}(b) = u_2 [F(b)-F(h(b))] + u_3 [1-F(b) + F(h(b)) - F(y(b))],
\]
where $h(b)$ and $y(b)$ are implicitly defined by the two equations in (\ref{ime}). 

\begin{claim} \label{ibm}
    If no commitment outcome is implementable, the sender's ex-ante payoff in a sender-preferred equilibrium, $\bar{V}(b)$, is increasing in $b$.
\end{claim}

\begin{proof}[Proof of \autoref{ibm}]
    Directly,
    \begin{equation} \label{drv}
    \bar{V}'(b) = \left(u_2-u_3\right)\left[f\left(b\right)-f(h) \frac{\mathrm{d} h}{\mathrm{d} b}\right]-u_3 f(y) \frac{\mathrm{d} y}{\mathrm{d} b}. 
    \end{equation}
    Using (\ref{ime}), by the implicit function theorem,
    \begin{align}
    \frac{\mathrm{d} h}{\mathrm{d} b} & = -\frac{\left(b-\gamma_2\right) f\left(b\right)}{\left(\gamma_2-h\right) f(h)}, \label{csa} \\
    \frac{\mathrm{d} y}{\mathrm{d} b} & = -\frac{\left(b-h\right)\left(\gamma_3-\gamma_2\right) f\left(b\right)}{\left(\gamma_2-h\right)\left(\gamma_3-y\right) f(y)}. \label{csb}
\end{align}
Plugging (\ref{csa}) and (\ref{csb}) into (\ref{drv}),
\begin{align*}
    \bar{V}'(b) & = \left(u_2-u_3\right)f(b)\left(1+ \frac{b-\gamma_2}{\gamma_2 - h}\right) + u_3 f(b) \frac{\left(b-h\right)\left(\gamma_3-\gamma_2\right)}{\left(\gamma_2-h\right)\left(\gamma_3-y\right)} \\
    & = \left[\left(u_2-u_3\right) \frac{b-h}{\gamma_2-h} + u_3 \frac{\left(\gamma_3-\gamma_2\right)\left(b-h\right)}{\left(\gamma_2-h\right)\left(\gamma_3-y\right)}\right] f\left(b\right) \\
    & = \left[u_2 \frac{b-h}{\gamma_2-h} - u_3 \frac{(b-h)(\gamma_2-y)}{(\gamma_2 - h)(\gamma_3 - y)} \right] f\left(b\right) \\
    & = \left[u_2 (\gamma_3 - y) - u_3 (\gamma_2 - y)\right] \frac{(b-h) f\left(b\right)}{(\gamma_2 - h)(\gamma_3 - y)},
\end{align*}
and we see that $\bar{V}'(b) \ge 0$ if and only if $u_2 (\gamma_3 - y) \ge u_3 (\gamma_2 - y)$. Then since no commitment solution is implementable, by \autoref{ntm}, (\ref{rie}) implies that $\bar{V}'(b) \ge 0$, and hence the sender's ex-ante payoff in a sender-preferred equilibrium is increasing in $b$. 
\end{proof}

\subsection{Remaining Omitted Proofs} \label{appendix:remaining-proofs}
\subsubsection{Proof of \autoref{claim:ternary-sp}}
    To solve for a sender-preferred equilibrium, we find a bi-pooling solution to the corresponding information design problem first, and check whether it is implementable using Corollary 1. If it is, a sender-preferred equilibrium is associated with a barely obedient bi-pooling partition $\mathcal{B}$ that is also associated with  the commitment solution.\footnote{One may wonder what if the bi-pooling solution found above is not implementable, but there exists another bi-pooling solution that is implementable. This can never happen when there are three actions: by \autoref{imc}, all bi-pooling solutions induce essentially the same bi-pooling partition.}

Now suppose that no commitment solution is implementable. By Theorem 2, there exist $y$, $z$, and $b$ such that $B_2 = [z,b]$, and $B_3 = [y, z] \cup [b,1]$.\footnote{If both $B_2$ and $B_3$ are intervals, just set $y = z$.} By \autoref{ntm}, there exist $y \ge 0$ and $z > 0$ such that
\begin{align*}
    \mathbb{E}\left[\omega \mid \omega \in\left[z, b\right]\right]=\gamma_2 \quad 
	\text{and} \quad \mathbb{E}\left[\omega \mid \omega \in\left[y, z\right] \cup\left[b, 1\right]\right]=\gamma_3.
\end{align*}
Consequently, $y$ and $z$ can be implicitly defined as functions of $b$, and hence the sender's ex-ante payoff can be parametrized by $b$, so long as $b \le \gamma_3$:
\[
\bar{V}(b) = u_2 [F(b)-F(z(b))] + u_3 [1-F(b) + F(z(b)) - F(y(b))].
\]
By \autoref{ibm}, $\bar{V}$ is increasing in $b$. Hence, the partition corresponding to the sender's preferred equilibrium can be found by setting $b = \gamma_3$, which yields the expression in the statement of the claim.

\subsubsection{Proof of \autoref{corollary:ternary-sufficient}}
Recall that $\gamma_1 := 0$ and $u_0$ is normalized to zero; hence when $n = 3$, Equation (4) in the main text reduces to
\begin{equation} \label{n3n}
	\frac{u_3-u_2}{\gamma_3-\gamma_2}>\frac{u_2}{\gamma_2 - \max\left\{0, h\left(\gamma_2 ; \gamma_3\right)\right\}} 
\end{equation}
And because $h(\gamma_2 ; \gamma_3) \ge 0$, the right-hand side of \eqref{n3n} further reduces to $u_2 /(\gamma_2-h\left(\gamma_2 ; \gamma_3\right))$. Then since $f$ is increasing, $\gamma_3-\gamma_2 \ge \gamma_2-h\left(\gamma_2 ; \gamma_3\right)$; thus, if $u_3 - u_2 > u_2$, or $u_3 > 2u_2$, \eqref{n3n} must hold. Consequently, by Proposition 2, every bi-pooling solution can be implemented. By \autoref{imc}, all bi-pooling solutions induce the same bi-pooling partition. Then because the set of bi-pooling solutions is the set of extreme points of the solution correspondence of the information design problem, all commitment solutions must be associated with the same bi-pooling partition. Thus, every commitment outcome is implementable.

\subsubsection{Proof of \autoref{msm}}
Let $\omega_q$ denote the cutoff quality that the buyer is indifferent between purchasing $q-1$ and $q$ units: it solves
\[
\omega_q U(q) - p q = \omega_q U(q-1) - p (q-1),
\]
so $\omega_q = p/[U(q) - U(q-1)]$. Letting $\omega_0 = 0$ and $\omega_{n+1} = 1$, the buyer buys $q \in \{0, 1, \ldots, n\}$ units of the product if and only if $\omega \in [\omega_{q}, \omega_{q+1}]$. If $P^U > 2A^U$, $U''(q)/[U'(q)]^2$ is strictly decreasing in $q$, and thus
\[
\omega_{q+1} - \omega_{q} = \frac{p}{U(q+1)-U(q)}-\frac{p}{U(q)-U(q-1)}
\]
is strictly decreasing in $q$. Since the seller's gain from the buyer purchasing one more unit is $p - c$, \autoref{msm} follows from \autoref{corollary:ternary-sufficient}.

\subsubsection{Proof of \autoref{claim:influencing-voters}}
If either $\alpha_3$ increases, or $\lambda_3$ increases, or both, $\gamma_3$ decreases. By Proposition 3, when no commitment outcome can be implemented, the sender's ex-ante payoff is given by 
\[
\bar{V}(\gamma_3) = u_2 [F(\gamma_3)-F(h(\gamma_3))] + u_3 [1-F(\gamma_3) + F(h(\gamma_3)) - F(y(\gamma_3))],
\]
where $h$ and $y$ are implicitly defined by
\[
\mathbb{E}\left[\omega \mid \omega \in\left[h, \gamma_3\right]\right]=\gamma_2 \quad 
	\text{and} \quad \mathbb{E}\left[\omega \mid \omega \in\left[y, h\right] \cup\left[\gamma_3, 1\right]\right]=\gamma_3.
\]
Now
\begin{equation} \label{csd1}
    \bar{V}'(\gamma_3) = \left(u_2-u_3\right)\left[f\left(\gamma_3\right)-f(h) \frac{\mathrm{d} h}{\mathrm{d} \gamma_3}\right]-u_3 f(y) \frac{\mathrm{d} y}{\mathrm{d} \gamma_3}.
\end{equation}

By the implicit function theorem, 
\begin{align}
    \frac{\mathrm{d} h}{\mathrm{d} \gamma_3} & = -\frac{\left(\gamma_3-\gamma_2\right) f\left(\gamma_3\right)}{\left(\gamma_2-h\right) f(h)}, \label{csd3} \\
    \frac{\mathrm{d} y}{\mathrm{d} \gamma_3} & = \frac{\left(\gamma_2-h\right)\left[1-f\left(\gamma_3\right)+F(h)-F(y)\right]}{\left(\gamma_2-h\right)\left(\gamma_3-y\right) f(y)}-\frac{\left(\gamma_3-h\right)\left(\gamma_3-\gamma_2\right) f\left(\gamma_3\right)}{\left(\gamma_2-h\right)\left(\gamma_3-y\right) f(y)}. \label{csd5} 
\end{align}
Plug (\ref{csd3}) and (\ref{csd5}) into (\ref{csd1}), 
\begin{align*}
    \bar{V}'(\gamma_3) & = -\frac{\left(\gamma_2-h\right)\left[1-f\left(\gamma_3\right)+F(h)-F(y)\right] u_3}{\left(\gamma_2-h\right)\left(\gamma_3-y\right)}-\frac{\left(\gamma_3-h\right)\left(\gamma_2-y\right) f\left(\gamma_3\right)}{\left(\gamma_2-h\right)\left(\gamma_3-y\right)} u_3 \\
    & \quad +\frac{\left(\gamma_3-h\right)\left(\gamma_3-y\right) f\left(\gamma_3\right)}{\left(\gamma_2-h\right)\left(\gamma_3-y\right)} u_2,
\end{align*}
whose sign is determined by
\begin{equation} \label{css}
-\left(\gamma_2-h\right)\left(1-F\left(\gamma_2\right)+F(h)-F(y)\right)-\left(\gamma_3-h\right) f\left(\gamma_3\right)\left[\left(\gamma_2-y\right) u_3-\left(\gamma_2-y\right) u_2\right].
\end{equation}
By \autoref{ntm}, if no comment outcome is implementable, it must be that $(\gamma_2-y) u_3 \le (\gamma_3 - y) u_2$. 
Consequently, the sign of the second term of (\ref{css}) must be positive, and the first term has a strictly negative sign. Hence as $\gamma_3$ decreases, the expert's ex-ante payoff in her preferred equilibrium strictly decreases if the second term is larger in absolute value, which establishes the statement.

\subsection{Equilibrium Refinement} \label{refinement}
It is natural to ask whether the equilibria of the disclosure game considered in this paper are credible in that they survive certain equilibrium refinements. We consider the following two equilibrium refinements:
\begin{itemize}
    \item The Never-a-Weak-Best-Response (NWBR) Criterion, proposed by \cite{chokreps1987}, is a strengthening of a few equilibrium refinements that are extensively used in the literature, which includes the Intuitive Criterion, D1, and D2.\footnote{Although closely related, this is not precisely the same as the NWBR property proposed by \cite{km86}. For this reason, \cite{fudenberg1991game} call it ``NWBR in signaling games.''}

    \item The Grossman-Perry-Farrell equilibrium, proposed by \cite{bertomeu2018verifiable}, is based on the perfect sequential equilibrium of \cite{grossman1986perfect} and neologism-proofness of \cite{farrell1993meaning}.
\end{itemize}
It can be shown that every PE of the disclosure game we study is a Grossman-Perry-Farrell equilibrium, and every PE outcome survives the NWBR Criterion. We use the term ``type'' instead of ``state'' henceforth to ease exposition.

\subsubsection{Never-a-Weak-Best-Response (NWBR) Criterion}
We introduce some notation first. For any $m \in \mathcal{C}$, let $MBR(m)$ denote the set of all mixed strategy best responses for the receiver to message $m$ for any belief $\beliefm(\cdot \mid m)$.\footnote{Note that $\beliefm$ must satisfy $\supp{\beliefm(\cdot \mid m)} \subseteq m$.} Moreover, let $v^*_\omega$ denote the equilibrium payoff of type $\omega$. Finally, for any equilibrium and an off-path message $m$, define
\[
 D(\omega,m) = \left\{\rho \in MBR(m) : v^*_\omega < \sum_{i \in N} u_i \rho_i \right\},
\]
and
\[
 D^0(\omega,m) = \left\{\rho \in MBR(m) : v^*_\omega = \sum_{i \in N} u_i \rho_i \right\};
\]
in words, $D(\omega,m)$ is the set of mixed strategy best responses that make type $\omega$ strictly prefer $m$ to her equilibrium message, and $D^0(\omega,m)$ is the set of mixed strategy best responses that make type $\omega$ exactly indifferent.

\begin{definition}
    An equilibrium $(\sigma, \tau, \beliefm)$ \textbf{survives the NWBR criterion} if for every $m \in \mathcal{C}$ and any $\omega, \omega' \in [0,1]$, $D^0(\omega',m) \subseteq \cup_{\omega \ne \omega'} D(\omega,m)$ implies that $\omega' \not \in \supp \beliefm \, (\cdot \mid m)$.
\end{definition}

\begin{claim}
    For every PE $(\sigma, \tau, \beliefm)$, there exists $\beliefm'$ such that $(\sigma, \tau, \beliefm')$ survives the NWBR criterion.
\end{claim}

\begin{proof}
    Fix a PE $(\sigma, \tau, \beliefm)$, and let $\mathcal{B}$ denote the associated partition. For every $m \not \in \mathcal{B}$, let 
    $\ell = \min\{i: m \cap B_i  \ne \varnothing\}$. Let $\beliefm'$ be such that $\beliefm'(\cdot \mid B_i) = \beliefm(\cdot \mid B_i)$ for all $i = 0, \ldots, n-1$, and for any $m \not \in \mathcal{B}$, let
    \begin{equation} \label{mbr}
        \beliefm' \, (\min m \cap B_\ell \mid m) = 1.
    \end{equation}
    By the definition of PE, the receiver never mixes on path, and hence for any $\omega' \in [0,1]$, $D^0(\omega',m) = \{\delta_k\}$ if and only if $\omega \in B_k \backslash \left(\cup_{i>k} B_i\right)$, where $\delta_k$ is the Dirac measure at action $k$. Furthermore, define
    \[M = \{k \in N : k \text{ is such that } m \cap B_k \ne \varnothing\};\]
    then the lowest action in $M$ is $\ell$. Now for any $\omega' \in [0,1]$, 
    \[\bigcup_{\omega \ne \omega'} D(\omega,m) = \{\rho \in \Delta (M) : \supp{\rho} \subseteq \{a_k, a_{k+1}\} \text{ with } k \ge \ell, \text{ and } \rho(a_{\ell}) < 1\}.\]
    Then $D^0(\omega',m) \subseteq \cup_{\omega \ne \omega'} D(\omega,m)$ if and only if $D^0(\omega',m) = \{\delta_k\}$ with $k > \ell$, which is in turn equivalent to $\omega' \in \cup_{k > \ell}\left(B_k \backslash \left(\cup_{i > k} B_i\right) \right)$. Then \eqref{mbr} implies that $\omega' \not \in \supp \beliefm' \, (\cdot \mid m)$. Consequently, $(\sigma, \tau, \beliefm')$ survives the NWBR criterion.
\end{proof}

\subsubsection{Grossman-Perry-Farrell Equilibrium}
\begin{definition}
    Fix a PE $(\sigma, \tau, \beliefm)$. Say that $m^* \in \mathcal{C}$ is a \textbf{self-signaling set} if\footnote{We could have defined a self-signaling set for any equilibrium, but doing that largely complicates the notation: here, $\sigma(\cdot)$ is a well-defined function that maps a state to a subset of the state space because $(\sigma, \tau, \beliefm)$ is an ORE.}
    \[m^* = \left\{\omega \in m^* : v\left(\mathbb{E}[\omega \,|\, \omega \in m^*]\right) > v\left(\mathbb{E}[\omega \,|\, \omega \in \sigma(\omega)]\right) \right\}.\]
    An PE is a \textbf{Grossman-Perry-Farrell equilibrium} if there does not exist a self-signaling set. 
\end{definition}

\begin{claim}
    Every PE is a Grossman-Perry-Farrell equilibrium.
\end{claim}

\begin{proof}
    Fix a PE $(\sigma, \tau, \beliefm)$, and let $\mathcal{B}$ denote the associated partition. Suppose there exists a self-signaling set $m^*$. Let
    \[\Bar{k}:= \max\left\{i \in N: \text{there exists } \omega \in m^* \text{ such that } \sigma(\omega) = B_i\right\}.\]
    Because $m^*$ is a self-signaling set, it must be that $\mathbb{E}[\omega \,|\, \omega \in m^*] \ge \gamma_{\bar{k}+1}$. But then revelation proofness of $\mathcal{B}$ implies that there must exist $\omega' \in m^*$ such that $\sigma(\omega') = B_j$ with $j \ge \bar{k}+1$, a contradiction. 
\end{proof}

\begin{singlespace}\small
\bibliographystyle{ecta}
\bibliography{wvi.bib}
\end{singlespace}

\end{document}